\documentclass[twocolumn,aps,pra,superscriptaddress,nofootinbib,notitlepage,showpacs,floatfix,longbibliography]{revtex4-2}
\usepackage{graphicx} % Required for inserting images
\usepackage{amsmath,amssymb,amsthm,dsfont,physics}
\usepackage[x11names]{xcolor}
\usepackage{mathtools}
\usepackage{braket}
\usepackage{enumitem}
\usepackage{amsthm}
\usepackage{empheq}
\usepackage{subfigure}
\usepackage{cancel}
\usepackage{bm}
\usepackage{mleftright}
\usepackage[normalem]{ulem}
\usepackage{relsize}
\usepackage[pdfstartview=FitH]{hyperref}
\hypersetup{
    colorlinks=true,       % false: boxed links; true: colored links
    linkcolor=red,          % color of internal links
    citecolor=magenta,        % color of links to bibliography
    filecolor=magenta,      % color of file links
    urlcolor=red,           % color of external links
    runcolor= blue
}
\usepackage[ruled]{algorithm}
\usepackage{algpseudocode}
\usepackage{multirow}
\usepackage{booktabs}
\usepackage{hhline}
\usepackage{tabularx}
\usepackage{adjustbox}
\usepackage[margin=1in,a4paper]{geometry}  % choose margins appropriately
\usepackage{array,booktabs,ragged2e}
\newcolumntype{R}[1]{>{\RaggedLeft\arraybackslash}p{#1}}
\newcommand{\specialcelll}[2][c]{%
  \begin{tabular}[#1]{@{}l@{}}#2\end{tabular}}

\makeatletter
\renewcommand*\env@matrix[1][*\c@MaxMatrixCols c]{%
  \hskip -\arraycolsep
  \let\@ifnextchar\new@ifnextchar
  \array{#1}}
\makeatother

\newtheorem{theorem}{Theorem}
\newtheorem*{theorem*}{Theorem}
\newtheorem{corollary}{Corollary}[theorem]

\newtheorem{definition}{Definition}

\newtheoremstyle{problemstyle}  % <name>
        {3pt}                                               % <space above>
        {3pt}                                               % <space below>
        {\normalfont}                               % <body font>
        {}                                                  % <indent amount}
        {\bfseries\itshape}                 % <theorem head font>
        {\normalfont\bfseries:}         % <punctuation after theorem head>
        {.5em}                                          % <space after theorem head>
        {}                                                  % <theorem head spec (can be left empty, meaning `normal')>
\theoremstyle{problemstyle}

\definecolor{beaublue}{rgb}{0.74, 0.83, 0.9}
\definecolor{non-photoblue}{rgb}{0.64, 0.87, 0.93}
\definecolor{paleaqua}{rgb}{0.74, 0.83, 0.9}
\definecolor{aliceblue}{rgb}{0.93, 0.97, 1.0}
\colorlet{lcfree}{Green3}
\colorlet{lcnorm}{Blue3}
\colorlet{lccong}{Red3}
\newcommand{\sket}[1]{| #1 \rangle}
  
\newcommand{\bsf}[1]{\bm{\mathsf{#1}}}
\newcommand{\bmsf}[1]{\bm{\mathsf{#1}}}
\newcommand{\bcal}[1]{\bm{\mathcal{#1}}}

\newcommand{\vertiii}[1]{{\| #1\|}}
\DeclarePairedDelimiter\pars{\lparen}{\rparen}
\usepackage{tikz}
\usepackage{tikz-cd}
\usepackage{verbatim}
\usepackage{hf-tikz}
\usepackage{adjustbox}
\usetikzlibrary{shapes,positioning,arrows,chains,fit,backgrounds,shapes.misc}
\usetikzlibrary{patterns,tikzmark}
\usetikzlibrary{matrix,decorations.pathreplacing,calc}
\usetikzlibrary{positioning, shapes.misc}
\pgfkeys{tikz/mymatrixenv/.style={decoration={brace},every left delimiter/.style={xshift=8pt},every right delimiter/.style={xshift=-8pt}}}
\pgfkeys{tikz/mymatrix/.style={matrix of math nodes,nodes in empty cells,left delimiter={[},right delimiter={]},inner sep=1pt,outer sep=1.5pt,column sep=8pt,row sep=8pt,nodes={minimum width=20pt,minimum height=10pt,anchor=center,inner sep=0pt,outer sep=0pt}}}
\pgfkeys{tikz/mymatrixbrace/.style={decorate,thick}}

\tikzset{style green/.style={
    set fill color=green!50!lime!60,draw opacity=0.4,
    set border color=green!50!lime!60,fill opacity=0.1,
  },
  style cyan/.style={
    set fill color=cyan!90!blue!60, draw opacity=0.4,
    set border color=blue!70!cyan!30,fill opacity=0.1,
  },
  style orange/.style={
    set fill color=orange!90, draw opacity=0.8,
    set border color=orange!90, fill opacity=0.3,
  },
  style brown/.style={
    set fill color=brown!70!orange!40, draw opacity=0.4,
    set border color=brown, fill opacity=0.3,
  },
  style purple/.style={
    set fill color=violet!90!pink!20, draw opacity=0.5,
    set border color=violet, fill opacity=0.3,    
  },
  kwad/.style={
    above left offset={-0.1,0.23},
    below right offset={0.10,-0.36},
    #1
  },
  pion/.style={
    above left offset={-0.07,0.2},
    below right offset={0.07,-0.32},
    #1
  },
  poz/.style={
    above left offset={-0.03,0.18},
    below right offset={0.03,-0.3},
    #1
  },set fill color/.code={\pgfkeysalso{fill=#1}},
  set border color/.style={draw=#1}
}

\begin{document}
\title{Analog classical simulation of closed quantum systems}
\author{Ka-Wa Yip}
\email{yipkawa@gmail.com}
\thanks{Now at Northeastern University.}
\affiliation{Zhejiang Lab, Hangzhou, 311121, China}
%\date{March 2023}
\begin{abstract}
We develop an analog classical simulation algorithm of noiseless quantum dynamics. By formulating the Schr\"{o}dinger equation into a linear system of real-valued ordinary differential equations (ODEs), the probability amplitudes of a complex state vector can be encoded in the continuous physical variables of an analog computer. Our algorithm reveals the full dynamics of complex probability amplitudes. Such real-time simulation is impossible in quantum simulation approaches without collapsing the state vector, and it is relatively computationally expensive for digital classical computers. 
For a real symmetric time-independent Hamiltonian, the ODEs may be solved by a simple analog mechanical device such as a one-dimensional spring-mass system. 
Since the underlying dynamics of quantum computers is governed by the Schr\"{o}dinger equation, our findings imply that analog computers can also perform quantum algorithms. We illustrate how to simulate the Schr\"{o}dinger equation in such a paradigm, with an application to  quantum approximate optimization algorithm. This may pave the way to emulate quantum algorithms with physical computing devices, including analog, continuous-time circuits. 
\end{abstract}
\maketitle
% \def\thefootnote{*}\footnotetext{Now at Northeastern University.}
% \def\thefootnote{\arabic{footnote}}
% \setcounter{footnote}{0} 
%%%%%%%%%%%%%%%%%%%%%%%%%%%%%%%%%%%%%%%%%%%%%%%
\section{Introduction}
Quantum simulation~\cite{RevModPhys.86.153,wiebe2011simulating,babbush2019quantum,cirac2012goals} of quantum systems has gained popularity over the last decade. 
The purpose of quantum simulation (also referred to as Hamiltonian simulation) is to use a quantum computer or simulator to approximate or simulate the exact quantum dynamics under the target quantum Hamiltonian. The target Hamiltonian\footnote{Henceforth, we will refer to quantum Hamiltonian, instead of classical Hamiltonian, as the Hamiltonian.} spans systems of interest across a range of fields from quantum chemistry~\cite{lidar1999} and many-body physics~\cite{raeisi2012} to high energy physics~\cite{preskill2018quantum}. There are two general quantum simulation approaches. One is digital quantum simulation (DQS), where a quantum circuit emulates the unitary evolution of a quantum system~\cite{lloyd1996universal}. The unitary operator is approximated using a sequence of quantum gates, employing techniques such as Trotterization~\cite{wiebe2011simulating, trotter_theory, aharonov2003adiabatic}, Taylor series truncation~\cite{berry2015simulating}, and qubitization~\cite{low2019hamiltonian}. DQS is well-suited for simulating local Hamiltonians. The second is analog quantum simulation (AQS), which simulates the target Hamiltonian by another~\cite{aharonov_et_al:LIPIcs:2018:10095}. This involves mapping the target Hamiltonian to the Hamiltonian in a well-controllable quantum simulator~\cite{probing51,manybody53}. For example, one can simulate the Hamiltonian dynamics of the Hubbard model with fermionic cold atoms in optical field~\cite{daley2022practical}. It is also a wide belief that quantum computers are tailor-made to simulate highly entangled quantum systems~\cite{preskill2018quantum}. 

Typically, discussions on classical computing machinery gravitate toward digital classical computers, which work with binary bits (0/1). 
The application of digital classical computers to simulate quantum systems is commonly referred to as digital classical simulation (DCS). 
In this digital approach, calculating the unitary operator $e^{-iHt}$ or solving the Schr\"{o}dinger equation over a time interval can be computationally expensive. 
Often, the challenges arise from the poor scalability of matrix exponentiation of the Hamiltonian and the use of numerical techniques for ordinary differential equations (ODEs), like Runge-Kutta methods, as the system size increases. 
While a vast literature exists on the topics of quantum simulations by digital and analog quantum computers, as well as by digital classical computers, much less is known for the topic of quantum simulations by analog classical computers. There is a gap in research regarding the capability of analog classical computers, i.e., devices that handle continuous, real-valued physical quantities without depending on quantum mechanics, to simulate quantum systems. The efficacy of such analog classical computers in mimicking quantum dynamics is also not well understood.
Driven by this gap in knowledge, we propose and investigate an alternative approach of using classical physical systems to simulate quantum systems: the analog classical simulation (ACS). 
Operating as a continuous time machine that mirrors natural physical processes~\cite{bournez2020computability, bournez2021survey,bournez2008survey}, the analog classical computer, and more broadly, the mathematical framework of the General Purpose Analog Computer (GPAC)~\cite{shannon1941mathematical}, can hypothetically sidestep some of the computational complexities associated with DCS, such as numerical integration of ODEs.

Our main contribution is an analog classical simulation algorithm for the Schr\"{o}dinger equation, which governs the dynamics of a closed quantum system. The algorithm leads to multiple discussions on the intricacies of its realization and the computational implications. We first detail a mapping from the forward problem of a time-dependent Schr\"{o}dinger equation to that of a system of real-valued ODEs. A time-dependent Schr\"{o}dinger equation is shown to be equivalent to two different systems of real-valued ODEs. We delve into one of them and illustrate how the Hamiltonian properties, such as time-dependence, can influence the structure of the corresponding system of ODEs. The global stability of the real ODEs can be proven to stem from the Hermitian nature of the Hamiltonian. 

Given the Hamiltonian and an initial state as inputs, the solutions to the real ODEs correspond to the probability amplitudes in a standard basis, e.g., computational,  without extra complexity in basis change. Therefore, by solving the real ODEs, the analog classical computer deterministically reveals the dynamics of quantum state's probability amplitudes. This feature of ACS stands out, particularly because DQS and AQS require numerous measurements (state vector collapses) and restarts to estimate the probability amplitudes at a given time. This allows a simulation longer than the decoherence timescale ($T_1/T_2$ time) posted by the quantum system, and a linear fast-forwarding of the quantum dynamics. Since the probabilities are known, one can also subsequently mimic quantum measurements to generate the output bitstrings or calculate the expectation values of quantum operators.
The tunability of parameters in an analog device also allows for simulations of general time-dependent Hamiltonians.
To bypass a detailed exploration of the engineering particulars of analog computing, we mostly illustrate its principles with a simple analog system comprised of two mechanical components: springs and masses.

A quantum algorithm that 
simulates the coupled classical oscillators was proposed~\cite{babbush2023exponential}. 
Although it primarily outputs macroscopic physical properties such as kinetic energy, rather than the full dynamics of the spring-mass system, the application of a quantum computer to simulate a classical mechanical system has profound implications. 
The underlying technique is the mapping of the classical dynamics governed by Newton's equation to the quantum dynamics. In contrast, our work employs encodings that translate a general time-dependent Schr\"{o}dinger equation into real ODEs, with the spring-mass system described by Newton's laws as a special case.

Although a simple spring-mass setup operates in continuous time and continuous space, its mechanical design poses certain limitations: positive masses and positive spring constants,  making it fundamentally a subset of GPAC. 
Still, it can simulate the Schrödinger equation for a large fraction of Hamiltonians. 
We introduce different encoding schemes of the spring mass system. By adjusting the spring constants, we show that a linear fast forwarding (speedup) of the dynamics is possible.

It is known that simulating time-independent Schr\"{o}dinger equation allows one to solve problems like glued trees~\cite{babbush2023exponential} and unstructured search~\cite{grover2002classical,patel2007coupled}. 
We also show that our analog classical simulation techniques enable the execution of quantum optimization algorithms such as quantum approximate optimization algorithm (QAOA), on an analog computer. 
% We also show that an analog classical computer can execute quantum optimization algorithms such as quantum approximate optimization algorithm (QAOA). 
Additionally, these real-valued ODEs are intrinsically connected to the dynamics of other physical computing methods, such as optical and neural computation~\cite{hasler2016opportunities, yon1958computer, mead1990neuromorphic, mead1989analog}. This paves the way to emulate quantum algorithms on other physical computing devices.

In Sec.~\ref{sec:pf}, we give an overview of the analog computer and formulate the simulation problem of general real-valued ODEs. In Sec.~\ref{sec:mapping} we unravel the Schr\"{o}dinger equation into a system of real ODEs, and illustrate how an analog computer can take advantage of this form of real ODEs to perform simulation.  In Sec.~\ref{sec:closed_algo} we provide an algorithmic implementation of analog classical simulation for a closed quantum system. 
In Sec.~\ref{sec:properties}, we delve into how different Hamiltonian structures successively lead to different properties of real ODEs, i.e., the Schr\"{o}dinger equation in real form. This serves as the foundation for the rest of the paper, outlining how to determine the simulator's capabilities.
In Sec.~\ref{sec:real_2nd_lds}, we focus on a particular analog simulator --- a second-order linear dynamical system --- 
and show that it can simulate any time-independent Hamiltonian. If the Hamiltonian is also real symmetric, the linear dynamical system reduces to two decoupled subsystems. Then in Sec.~\ref{sec:springmass} we focus on a 1-D spring-mass system, a specific subset of the second-order linear dynamical system, where the stiffness matrix~\cite{gladwell_book} has certain constraints posed by the positive spring constants. We present multiple schemes to encode Hamiltonians into this simple system.
In Sec.~\ref{sec:speedup}, we show how the scaling of spring constants bring forth the fast-forwarding of dynamics, relative to the inherent quantum timescale. We next show in Sec.~\ref{sec:qaoa} how our ACS approach can be applied to the quantum algorithm, with QAOA as an example. Additional technical details and proofs can be found in the Appendices.

\section{Problem formulation}
\label{sec:pf}
Many physical phenomena, such as electrical and mechanical processes, are described by ODEs over real numbers. An analog computer is designed to solve problems modeled by these differential equations, using the continuous dynamical properties of physical components that operate on similar working principles~\cite{siegelmann1998analog}.

Claude Shannon proposed the General-Purpose Analog Computer (GPAC)~\cite{shannon1941mathematical,bournez2006general} as a mathematical abstraction of the differential analyzer~\cite{thomson1876vi}, which is itself a mechanical analog computer. The GPAC can be presented in terms of polynomial ordinary differential equations~\cite{bournez2016computing}. It can also be regarded as the continuous-time counterparts of Turing machines, which are models for discrete time, in many respects~\cite{bournez2021survey}. For the purpose of simulating closed quantum systems, in general we want to use an analog classical computer or a GPAC to solve the following linear system of second-order real-valued ODEs:
\begin{equation}
\label{eq:2ndorderODE_1}
    \bm{\mathsf{A}}(t)\ddot{\pmb{y}}(t) + \bm{\mathsf{B}}(t)\dot{\pmb{y}}(t) + \bm{\mathsf{C}}(t)\pmb{y}(t) = \pmb{d}(t) \,.
\end{equation}
The overdot denotes derivative with respect to time and the bold symbol indicates that the object is a matrix or a vector. All components of Eq.~\ref{eq:2ndorderODE_1} are real. For generality, we also allow them to be time-dependent. By simulating Eq.~\ref{eq:2ndorderODE_1} from $t=0$ to $t=T$, with the initial conditions $\pmb{y}(0)$ and $\dot{\pmb{y}}(0)$, one can obtain information about the second-order dynamic system as defined by the equation at $t=T$. 

\section{Mapping the Schr\"{o}dinger equation to a linear system of real-valued ODEs}
\label{sec:mapping}
The Schr\"{o}dinger equation, in the form commonly encountered in quantum computing with a time-dependent Hamiltonian $\bm{\mathsf{H}}(t)$ is
\begin{align}
\label{eq:Schr}
 \ket{\dot{\psi}(t)} &= -i\bsf{H}(t) \sket{\psi(t)} \,,
\end{align}
where $\ket{\psi(t)} \in \mathbb C^{N\times 1}$ is the state of %an $(N+M)$-dimensional 
a quantum system at time $t$, and $N = 2^n$ for $n$ qubits. The norm $\|\psi(t)\| = 1$, for all $t \geq 0$. In Eq.~\ref{eq:Schr} we have set the reduced Planck constant $\hbar = 1$. Starting at time $t=0$, the state at time $t=T$ is given by
\begin{equation}
\ket{\psi(0+T)} = \bsf{U}(T, 0)\ket{\psi(0)} \,,
\end{equation}
where 
\begin{equation}
\label{eq:unitary}
\bsf{U}(T, 0) = {\cal T}\exp\left[-i\int_{0}^{T}\bsf{H}(t')dt'\right] \,.
\end{equation}
The time-ordering operator for the matrix exponential is denoted by $\cal T$. One main bottleneck of digital classical simulation for quantum systems is that the unitary $\bsf{U}(T, 0)$ is hard to compute. Note that Eq.~\ref{eq:Schr} is a linear system of complex-valued ODE. 

Our initial step involves converting the Schr\"{o}dinger equation (Eq.~\ref{eq:Schr}) into a system of homogeneous real-valued ODEs, in the form of Eq.~\ref{eq:2ndorderODE_1}. To accomplish this, there are two strategies. One converts the equation into a system of first-order real-valued ODEs, with the procedure provided in the Appendix~\ref{app:firstorder}. The other one, which we will detail here, involves transforming it into a system of second-order real-valued ODEs. 

Let us differentiate the Schr\"{o}dinger equation one more time. 
Starting from Eq.~\ref{eq:Schr}, we have
\begin{align}
\label{eq:schrodinger}
 \ket{\ddot{\psi}(t)}  &= -i {\dot{\bm{\mathsf{H}}}(t)} \sket{\psi(t)} -i{\bm{\mathsf{H}}(t)} \sket{\dot{\psi}(t)} \nonumber\\
 &= -i {\dot{\bm{\mathsf{H}}}}(t) \ket{\psi(t)} -i{\bm{\mathsf{H}}}(t) (-i {\bm{\mathsf{H}}(t)} \ket{\psi(t)}) \nonumber\\
 &= -i {\dot{\bm{\mathsf{H}}}}(t) \ket{\psi(t)} - \bm{\mathsf{H}}(t)^2 \ket{\psi(t)} \,.
\end{align}
Define $\bm{\mathsf{K}}(t) = i\dot{\bm{\mathsf{H}}}(t) + \bm{\mathsf{H}}(t)^2$; we then have 
\begin{equation}
\label{eq:2ndschro}
\ket{\ddot{\psi}(t)} = - \bm{\mathsf{K}}(t)|\psi(t)\rangle \,.   
\end{equation}

Eq.~\ref{eq:2ndschro} bears resemblance to the so-called second-order ``real"
Schr\"odinger equation~\cite{schrodinger2003collected,chen1989derivation,chen1991concerning,keepingitreal}, which is the first-published version of the Schr\"odinger equation~\cite{schrodinger2003collected}. However, in these references (see, for example, Eq.~1 of~\cite{keepingitreal}), the ``real" Schr\"odinger equation has $\bm{\mathsf{K}}(t)$ as a time-independent matrix with only real entries. Moreover, for many reasons (e.g.,~\cite{yang2013square,hestenes1993kinematic}), the current version of the Schr\"odinger equation has been extended into the complex space --- i.e., $\bm{\mathsf{H}}(t)$ and $\sket{\psi(t)}$ can be complex. Therefore, Eq.~\ref{eq:2ndschro} is not the historically ``real" Schr\"odinger equation. 

If $\bsf{H}(t)$ is time-independent, i.e., $\dot{\bsf{H}}(t) = 0$ for all time $t$, then the additional differentiation removes the imaginary unit $i$ prefactor of the Hamiltonian in Eq.~\ref{eq:Schr}. 
Effectively, $\bsf{K}$ is Hermitian\footnote{In this manuscript, if time dependence is not explicitly indicated, then it is understood that the operators are time-independent.}, a fact that will be valuable in later discussions. 

\subsection{Initial Value Problem}
\label{sec:eqn_equiv}
Every solution to the Schr\"odinger equation (Eq.~\ref{eq:Schr}) satisfies Eq.~\ref{eq:2ndschro}. On the other hand, the solutions to Eq.~\ref{eq:2ndschro} are the same as those of the Eq.~\ref{eq:Schr}, if the initial values $\ket{\psi(0)}$ and $\ket{\dot{\psi}(0)}$ for the two equations coincide.
Therefore, the corresponding initial value problem (IVP) of Eq.~\ref{eq:2ndschro} is as follows:
\begin{equation}
\label{eq:ivp}
\begin{cases}
\ket{\ddot{\psi}(t)} = - \bmsf{K}(t)\sket{\psi(t)} \,;\\
\ket{\dot{\psi}(0)} = - i\bsf{H}(0)\sket{\psi(0)} \,.
\end{cases}
\end{equation}
A detailed derivation is given in Appendix~\ref{app:proof_of_equivalence}.

Hence we only need to compute the Schr\"odinger equation (Eq.~\ref{eq:Schr}) once at the beginning of simulation. 

\subsection{Decomplexification}
We are in a position to turn the IVP  Eq.~\ref{eq:ivp} real. This process involves moving to a higher dimension, specifically from $N=2^n$ to $2N=2^{n+1}$.
Let $\Re{\bm{\mathsf{K}}(t)}$ and $\Im{\bm{\mathsf{K}}(t)}$ denote the real and imaginary parts, respectively, of $\bm{\mathsf{K}}(t)$ at time $t$, i.e., \[\bsf{K}(t) = \Re{\bm{\mathsf{K}}(t)} + i\Im{\bm{\mathsf{K}}(t)}.\] 
Define $\bcal{K}(t)$ as a matrix in $\mathbb{R}^{2N \times 2N}$, obtained by transforming $\bsf{K}(t)$ through a 
mapping as follows:
\begin{equation}
\bsf{K}(t) \mapsto
\bcal{K}(t) = 
\left[\begin{array}{cc}
\Re{\bm{\mathsf{K}}(t)} & -\Im{\bm{\mathsf{K}}(t)} \\
\Im{\bm{\mathsf{K}}(t)} & \Re{\bm{\mathsf{K}}(t)}
\end{array}\right] \,.
\end{equation}
The block matrix $\bcal{K}(t)$ is a real representation of $\bsf{K}(t)$~\cite{horn2012matrix}.
To complete the decomplexification of Eq.~\ref{eq:2ndschro}, we also turn the state vector $\ket{\psi(t)}$ in $\mathbb{C}^{N \times 1}$ to a real state vector $\ket{\varphi(t)}$ in $\mathbb{R}^{2N \times 1}$:
\begin{equation}
    \ket{\varphi(t)} = \begin{bmatrix} \Re{\ket{\psi(t)}} \\ \Im{\ket{\psi(t)}} 
\end{bmatrix}\,,
\end{equation}
where $\Re{\ket{\psi(t)}}$ and $\Im{\ket{\psi(t)}}$ are the real and imaginary parts of the state vector $\ket{\psi(t)}$ at time $t$, respectively. Eq.~\ref{eq:2ndschro} can be turned into the following real form:
\begin{equation}
\label{eq:realform}
    \sket{\ddot{\varphi}(t)} = -\bcal{K}(t)\sket{\varphi(t)}\,.
\end{equation}
An illustration of different approaches to compute $\bcal{K}(t)$ and a short proof of the real representation are given in Appendix~\ref{app:square} and~\ref{app:proof_of_equivalence_real}.

Likewise, the second equation of the IVP (Eq.~\ref{eq:ivp}) can be ``decomplexified" as below:
\begin{align}
\label{Eq:Schr_real1st_initial}
\ket{\dot{\varphi}(0)}
&=-\begin{bmatrix}
\Re{i\bsf{H}(0)} & -\Im{i\bsf{H}(0)} \\
\Im{i\bsf{H}(0)} & \Re{i\bsf{H}(0)}
\end{bmatrix}\sket{\varphi(0)} \nonumber\\
&= -\begin{bmatrix}
-\Im{\bsf{H}(0)} & -\Re{\bsf{H}(0)} \\
\Re{\bsf{H}(0)} & -\Im{\bsf{H}(0)}
\end{bmatrix}\sket{\varphi(0)}\,. \nonumber\\
\end{align}
Therefore, Eq.~\ref{eq:ivp} can be turned into the following real form:
\begin{equation}
\label{eq:ivp_real}
\begin{cases}
\sket{\ddot{\varphi}(t)} = -\bcal{K}(t)\sket{\varphi(t)} \,;\\
\ket{\dot{\varphi}(0)} = 
\small-\hspace{-0.5ex}\begin{bmatrix}
-\Im{\bsf{H}(0)} & \hspace{-0.8ex}-\Re{\bsf{H}(0)} \\
\Re{\bsf{H}(0)} & \hspace{-0.8ex} -\Im{\bsf{H}(0)}
\end{bmatrix}
\hspace{-1ex}
\begin{bmatrix}
\Re{\ket{\psi(0)}} \\
\Im{\ket{\psi(0)}}
\end{bmatrix}.
\end{cases}
\end{equation}

\section{Analog classical simulation algorithm for closed quantum system}
\label{sec:closed_algo}
We summarize Sec.~\ref{sec:mapping} by illustrating the analog classical algorithm for the Schr\"odinger equation in Alg.~\ref{alg:sim1}. The goal is to simulate the Schr\"{o}dinger equation for a duration of time $T$. The algorithm takes as inputs the initial state and the Hamiltonian; if the Hamiltonian is time-dependent, it requires its values at all times.

\begin{algorithm}[H]
\caption{\label{alg:sim1} Analog classical simulation for the Schr\"{o}dinger equation}\vspace{1.5pt}
\textbf{Input} $\ket{\psi(0)}$, $\bsf{H}(t) \,\,\,\,\, \forall t \in [0, T]$ \\[0pt]
\begin{algorithmic}[1]
\State Convert $\ket{\psi(0)}$ to  $\ket{\varphi(0)} = \begin{bmatrix} 
\Re{\ket{\psi(0)}} \\ 
\Im{\ket{\psi(0)}} 
\end{bmatrix}$ 
\State Set $\ket{\dot{\varphi}(0)}$ to \vspace{1.3pt} 
\newline 
$\small\begin{bmatrix}\Im{\bmsf{H}(0)}\Re{\ket{\psi(0)}} + \Re{\bmsf{H}(0)}\Im{\ket{\psi(0)}}\\
-\Re{\bmsf{H}(0)}\Re{\ket{\psi(0)}} + \Im{\bmsf{H}(0)}\Im{\ket{\psi(0)}} 
\end{bmatrix}$
\vspace{1pt}
\State Encode the initial configuration of the analog computer with $\ket{\varphi(0)}$ and $\ket{\dot{\varphi}(0)}$
\State Calculate $\bsf{K}(t) = i\dot{\bm{\mathsf{H}}}(t) + \bm{\mathsf{H}}(t)^2$
\vspace{1pt}
\State Map $\bsf{K}(t)$ to $\bcal{K}(t) = \left[\begin{array}{cc}
\Re{\bm{\mathsf{K}}(t)} & -\Im{\bm{\mathsf{K}}(t)} \\
\Im{\bm{\mathsf{K}}(t)} & \Re{\bm{\mathsf{K}}(t)}
\end{array}\right]$ 
\State Evolve $\ket{\varphi(0)}$ according to the following ODE: 
\begin{equation*}
    \sket{\ddot{\varphi}(t)} = -\bm{\mathcal{K}}(t)\sket{\varphi(t)}  \,,
\end{equation*}
from time $t=0$ to $t=T$
\State Measure and decode $\sket{\varphi(T)}$
\end{algorithmic}
\textbf{Result} $\sket{\varphi(T)}$
\vspace{3pt}
If necessary, obtain $\sket{\psi(T)}$ by \[\sket{\psi(T)} = \sket{\varphi(T)}_{0<j\leq 2^{n}} + i\sket{\varphi(T)}_{2^{n}<j\leq 2\cdot 2^{n}} \]
\end{algorithm}

\section{Properties of \texorpdfstring{$\bsf{K}(t)$}{bsfKt} and \texorpdfstring{$\bcal{K}(t)$}{bcalKt}}
\label{sec:properties}
Starting from the Schr\"{o}dinger equation we derive new matrices $\bsf{K}(t)$ and $\bcal{K}(t)$. This section presents various classes of properties associated with these new operators, and how they are related to the structure of the Hamiltonian $\bsf{H}(t)$. We summarize the relationship between the three operators in Table.~\ref{tab:K_vs_stiffness_simple}. 

\begin{table*}[htb!]
\caption{\label{tab:K_vs_stiffness_simple}%
A table illustrating the relationship between the Hamiltonian $\bsf{H}(t)$, $\bsf{K}(t)$ and $\bcal{K}(t)$.
}
  \centering
  \begin{ruledtabular}
  \begin{tabular}{@{}lll@{}}
  %\toprule
    $\bsf{H}(t)$  & $\bsf{K}(t)$  &$\bcal{K}(t)$  \\
    Time-dependent   &  non-Hermitian   & Real  \\
    \midrule
    $\bsf{H}$  & $\bsf{K}$  &$\bcal{K}$ \\ %\midrule 
    Time-independent \quad\quad\quad\quad\quad\quad\quad\quad\quad\quad\quad\quad\quad & \specialcelll{Hermitian\\Positive semidefinite}  & \specialcelll{Real symmetric\\Positive semidefinite}  \\ \cmidrule{2-3}
    \quad\quad\quad -- Real symmetric & Real symmetric & Block diagonal 
    \adjustbox{valign=c}{
    \begin{tikzpicture} [every left delimiter/.style={black,xshift=1ex},
   every right delimiter/.style={black, xshift=-1ex}]
      \matrix (m) at (0,0) [matrix of math nodes,
      left delimiter = {[},
      right delimiter = {]}, 
      nodes={white,fill=black!20,minimum size=3mm,text height=0.6ex, text width=0.7ex, text depth=0ex}]
      {
        \node {}; & \node[fill=blue!0]{};\\
        \node[fill=blue!0]{}; & \node{};\\
      };
    \end{tikzpicture}
    }
    \\ \cmidrule{2-3}
    \quad\quad\quad -- ZM-matrix & Z-matrix & Block diagonal Z-matrix 
        \adjustbox{valign=c}{
    \begin{tikzpicture}[every left delimiter/.style={black,xshift=1ex},
   every right delimiter/.style={black, xshift=-1ex}]
      \matrix (m1) at (0,0) [matrix of math nodes,
      left delimiter = {[},
      right delimiter = {]},
      nodes={white,fill=black!20,minimum size=1mm,text height=1.3ex, text depth=0}]
      {
        \node {\tiny Z}; & \node[black,fill=blue!0]{};\\
        \node[black,fill=blue!0]{}; & \node{\tiny Z};\\
      };
    \end{tikzpicture}
    }
    \\ 
  \end{tabular}
  \end{ruledtabular}
\end{table*}

\subsection{Properties of \texorpdfstring{$\bm{\mathsf{K}}(t)$}{Kt}}
\label{sec:properties_of_k}
We elaborate on the properties of $\bsf{K}(t)$, which also explain some of the motivations for the additional differentiation of the Schr\"odinger equation in Sec.~\ref{sec:mapping}. It is easy to deduce that $\dot{\bsf{H}}(t)$ and $\bsf{H}(t)^2$ are Hermitian due to the inherent Hermitian nature of $\bsf{H}(t)$. Therefore, $\bsf{K}(t) = i\dot{\bsf{H}}(t) + \bsf{H}(t)^2$ is generally non-Hermitian. Interesting properties of $\bsf{K}(t)$ happen when the operators are time-independent --- i.e., $\bsf{K} = \bsf{H}^2$ --- or when the Hamiltonian is slowly varying compared to its square norm:
$\vertiii{\dot{\bm{\mathsf{H}}}(t)} \ll \vertiii{\bm{\mathsf{H}}(t)^2}$.

\subsubsection{Hermitian \texorpdfstring{$\bsf{K}$}{K}}
\label{sec:hermitian_k}
At time $t$, the matrix $\bsf{K}(t) = i\dot{\bsf{H}}(t) + \bsf{H}(t)^2$ is Hermitian if and only if $\dot{\bsf{H}}(t) = 0$. Therefore, a sufficient condition for $\bsf{K}(t)$ to be Hermitian is that the Hamiltonian is time-independent; i.e., $\bsf{H}(t) = \bsf{H}$ for all $t$. 

\subsubsection{Real symmetric \texorpdfstring{$\bsf{K}$}{K}}
If $\bsf{H}$ is real symmetric, then $\bsf{K} = \bsf{H}^2$ is also real symmetric. This occurs in the case of time-reversal invariant Hamiltonian. There are additional scenarios where $\bsf{K}$ can be real symmetric. For example, when $\Re{\bsf{H}}\Im{\bsf{H}}$ is real symmetric. This condition is trivially met when $\bsf{H}$ is either itself real symmetric or purely imaginary. These possibilities are detailed in Appendix~\ref{app:real_symmetric_k}.

\subsubsection{The matrix \texorpdfstring{$\bsf{K}$}{K} is a \texorpdfstring{Z-matrix}{z-matrix}}
\label{sec:zm_k}
We introduce two matrix definitions. For more on the classification of matrix families, see for example Refs.~\texorpdfstring{\cite{azimzadeh2019fast,power_zmatrix,bermanbook}}{citeazimzadeh2019fast}.
\begin{definition} 
A Z-matrix is a real matrix whose off-diagonal entries are nonpositive.
\end{definition}
\begin{definition}
A ZM-matrix is a matrix whose every power is a Z-matrix.
\end{definition}
Hence, if $\bsf{H}$ is a ZM-matrix, then its square --- $\bsf{K}$ --- is a Z-matrix. A Z-matrix is also sometimes called a stoquastic operator~\cite{bravyi2008,marvian2019computational,klassen2019two, choi2021essentiality}.

\subsubsection{Positive semidefinite \texorpdfstring{$\bsf{K}$}{K}}
\label{sec:positivesemidefinite_K}
Since the eigenvalues of $\bsf{H}(t)$ are real, the eigenvalues of $\bsf{H}(t)^2$ are nonnegative. Therefore, for time-independent $\bsf{H}$, the Hermitian matrix $\bsf{K}$ is necessarily positive semidefinite. Positive semidefiniteness of $\bsf{K}$ also implies that it has nonnegative diagonal entries. Such matrices with nonpositive off-diagonal and nonnegative diagonal entries often arise in dynamical systems of many fields.

It is also well known that a scalar multiple $\alpha$ of an identity operator can be added to the Hamiltonian ($\bsf{H}(t) \rightarrow \bsf{H}(t) + \alpha\bsf{I}$) without changing the dynamics of the expectation values. Perturbing\footnote{This can be done without computing the eigenvalues of $\bsf{H}(t)$. Adding an arbitrary infinitesimal $\epsilon\bsf{I}$ to $\bsf{H}(t)$ perturbs the whole eigenspectrum by $+\epsilon$, effectively removing all zero eigenvalues. 
} $\bsf{H}$ in this manner can make its eigenvalues all nonzero. This in turn makes $\bsf{K}$ positive definite. A real symmetric positive definite Z-matrix is also called a Stieltjes matrix~\cite{bermanbook}. 

\subsection{Properties of \texorpdfstring{$\bcal{K}(t)$}{Kt}}
The real matrix $\bcal{K}(t)$ in the linear system of ODE possesses certain unique properties, which are dependent on the properties of $\bsf{K}(t)$. 
Each group of properties will generate its own unique set of dynamics. 

\subsubsection{Real symmetric \texorpdfstring{$\bcal{K}$}{K}}
\label{sec:real_symmetric_kcal}
If $\bsf{K}$ is Hermitian, then $\bcal{K}$ is real symmetric. A simple proof is given in Appendix~\ref{app:simpleproof_2}.

\subsubsection{Block diagonal \texorpdfstring{$\bcal{K}$}{K}}
\label{sec:block_diagonal_kcal}
If $\bsf{K}$ is real, then $\bcal{K}$ is block diagonal.

\subsubsection{The matrix \texorpdfstring{$\bcal{K}$}{K} is a block diagonal Z-matrix}
\label{sec:zm_kcal}
If $\bsf{K}$ is a Z-matrix, then $\bcal{K}$ is a block diagonal Z-matrix. 

\subsubsection{Positive semidefinite $\bcal{K}$}
\label{sec:positivesemidefinite_Kcal}
If $\bsf{K}$ is positive semidefinite, then $\bcal{K}$ is positive semidefinite. A simple proof is given in Appendix~\ref{app:simpleproof_3}.

\section{Real second-order linear dynamical system}
\label{sec:real_2nd_lds}
Recall that the second-order real-valued linear ODE of Eq.~\ref{eq:2ndorderODE_1} is
\begin{equation*}
    \bm{\mathsf{A}}(t)\ddot{\pmb{y}}(t) + \bm{\mathsf{B}}(t)\dot{\pmb{y}}(t) + \bm{\mathsf{C}}(t)\pmb{y}(t) = \pmb{d}(t) \,. \tag{\ref{eq:2ndorderODE_1}}
\end{equation*}
One can observe that the Schr\"odinger equation in real form:
\begin{equation*}
    \sket{\ddot{\varphi}(t)} = -\bcal{K}(t)\sket{\varphi(t)}\,, \tag{\ref{eq:realform}}
\end{equation*}
where $\bcal{K}(t) \in \mathbb R^{2N\times 2N}$, is a special case of 
Eq.~\ref{eq:2ndorderODE_1}.

We will further illustrate how, for a time-independent $\bsf{H}$, the Schr\"odinger equation in real form (Eq.~\ref{eq:realform}) can be analogously simulated by a second-order linear dynamical system, rather than a GPAC or more intricated analog computers. The dynamics of a second-order linear dynamical system is described by the following equation:
\begin{equation}
\label{eq:2ndorderODE_3}
    \bm{\mathsf{M}}\ddot{\pmb{x}}(t) + \bm{\mathsf{D}}\dot{\pmb{x}}(t) + \bm{\mathsf{S}}\pmb{x}(t) = \pmb{f}(t)\,.
\end{equation}
The matrices $\bsf{M}$, $\bsf{D}$, and $\bsf{S}$ are real symmetric and positive semidefinite~\cite{kawano2013decoupling,kawano2018decoupling}, representing the mass, damping, and stiffness, respectively.
The real vectors $\pmb{x}(t)$ and $\pmb{f}(t)$ specify the system's generalized coordinates and the external driving. 

\begin{theorem}
\label{thm:closedgeneral}
    A closed quantum system with a time-independent Hamiltonian $\bsf{H}$ can be analogously simulated with a second-order linear dynamical system. 
\end{theorem}
\begin{proof}
\label{pf:closedgeneral}
Since $\bsf{H}$ is time-independent, it follows that $\bsf{K}$ is Hermitian (see Sec.~\ref{sec:hermitian_k}). Furthermore, the $2N \times 2N$ matrix $\bcal{K}$ is real symmetric (see Sec.~\ref{sec:real_symmetric_kcal} or Table.~\ref{tab:K_vs_stiffness_simple}). From Sec.~\ref{sec:positivesemidefinite_K}, we know that $\bsf{K}$ is positive semidefinite, whence positive semidefiniteness of $\bcal{K}$ follows (see Sec.~\ref{sec:positivesemidefinite_Kcal}). Since $\bcal{K}$ is real symmetric and positive semidefinite, it can be realized through the stiffness matrix $\bsf{S}$, in a $2N \times 2N$ second-order linear dynamical system (Eq.~\ref{eq:2ndorderODE_3}), with $\bsf{M} = \bsf{I}$ and $\bsf{D} = \bsf{0}$.
\end{proof}

The procedure of the analog classical simulation through a second-order linear dynamical system is similiar to Alg.~\ref{alg:sim1}. In essence, the Schr\"odinger equation in real form:
\begin{equation*}
    \sket{\ddot{\varphi}(t)} + \bcal{K}\sket{\varphi(t)} = 0\,, \tag{\ref{eq:realform}}
\end{equation*}
is replaced with 
\begin{equation}
\label{eq:2ndorderODE_spring_new}
    \ddot{\pmb{x}}(t) + \bsf{S}\pmb{x}(t) = 0\,,
\end{equation}
for example, by letting $\pmb{x}(t) = d\ket{\varphi(t)}$ for some real constant such that $\max||\pmb{x}(t)|| = d$. Depending on the physical constraint of the linear dynamical system, one can also multiply the stiffness matrix $\bsf{S}$ with another positive constant to achieve time contraction or dilation of the quantum simulation.

\subsection{Decoupled second-order linear dynamical system}
\label{sec:decoupled2nd}
Furthermore, if $\bsf{H}$ is real symmetric, one can have the following observation:
\begin{corollary}
\label{thm:closedgeneral_2}
    If $\bsf{H}$ is real symmetric, the quantum system can be analogously simulated with a decoupled second-order linear dynamical system. 
\end{corollary}
\begin{proof}
A decoupled system can be separated into two subsystems such that there is no stiffness connection between the two.
It suffices to show that if $\bsf{H}$ is real symmetric, the corresponding stiffness matrix $\bsf{S}$ in Eq.~\ref{eq:2ndorderODE_spring_new} can be written as $\bsf{S}=\bsf{S}_1 \oplus \bsf{S}_2$. 

If $\bsf{H}$ is real symmetric, $\bcal{K}$ is block diagonal (see Sec.~\ref{sec:block_diagonal_kcal} or Table.~\ref{tab:K_vs_stiffness_simple}):
\begin{equation*}
    \bcal{K} =  \Re{\bsf{K}} \oplus \Re{\bsf{K}}  = \bsf{H}^2 \oplus \bsf{H}^2 \,.
\end{equation*}
Therefore, one can set $\bsf{S}_1, \bsf{S}_2 = \bsf{H}^2$, or scale them both by a constant if the simulation timescale need not exactly match that of the quantum system's evolution. The dynamics of a quantum system with a real symmetric $\bsf{H}$ is thus fully captured by Eq.~\ref{eq:2ndorderODE_spring_new} with $\bsf{S}=\bsf{S}_1 \oplus \bsf{S}_2$.
\end{proof}
Here are some observations for a real symmetric $\bsf{H}$. The Schr\"odinger equation in real form (Eq.~\ref{eq:realform}) describes the intertwining between the real and imaginary part of the wavefunction $\ket{\psi(t)}$ of the Schr\"odinger equation. For real symmetric $\bsf{H}$, the acceleration of the real, $\Re{\ket{\psi(t)}}$, and imaginary, $\Im{\ket{\psi(t)}}$, parts of probability amplitudes at time $t$ depends only on their respective values at the same time. 
Specifically, 
\begin{align*}
    \Re\{\ket{\ddot{\psi}(t)}\} = -\bsf{H}^2\Re{\ket{\psi(t)}}\,, \nonumber\\
    \Im\{\ket{\ddot{\psi}(t)}\} = -\bsf{H}^2\Im{\ket{\psi(t)}}\,.
\end{align*}

On the other hand, from Eq.~\ref{Eq:Schr_real1st_initial} (See also Eq.~\ref{Eq:Schr_real1st_2} in the Appendix), we have
\begin{equation}
\ket{\dot{\varphi}(t)} = -\begin{bmatrix}
\bsf{0} & -\bsf{H}\\
\bsf{H} & \bsf{0}
\end{bmatrix}\sket{\varphi(t)} \,.
\end{equation}
Therefore, the velocities of the real and imaginary parts of $\ket{\psi(t)}$ depend on the imaginary and real parts of $\ket{\psi(t)}$, respectively. Specifically, 
\begin{align*}
\Re\{\ket{\dot{\psi}(t)}\}  &= \bsf{H}\Im{\ket{\psi(t)}}\,, \nonumber\\
\Im\{\ket{\dot{\psi}(t)}\}  &= -\bsf{H}\Re{\ket{\psi(t)}}\,.
\end{align*}
This determines the initial, or ``kickoff", velocities of the two decoupled subsystems of the linear dynamical system, which are related to the initial displacements of the counterparts.

\subsection{Spring-mass system}
\label{sec:springmass}
As an example, we show how quantum dynamics can be simulated via a one-dimensional spring-mass system. This classical system is composed of $2N$ masses. The spring-mass connection itself can be represented as a simple graph, i.e., without multiple springs between masses or loops. Therefore, counting the wall, there are at most $N(2N+1)$ springs. Fig.~\ref{fig:massspring_12_a} is an illustration of the spring-mass setup, where the spring constants and masses are tunable parameters. In many cases of analog classical simulation, the spring-mass system is decoupled. A decoupled spring-mass system is also composed of $2N$ masses --- but at most $N(N+1)$ springs.

An unforced and undamped spring-mass system is a special case of the second-order linear dynamical system, where Eq.~\ref{eq:2ndorderODE_3} reduces to
\begin{equation}
\label{eq:2ndorderODE_spring}
    \bsf{M}\ddot{\pmb{x}}(t) + \bsf{S}\pmb{x}(t) = 0\,.
\end{equation}
In this case, the mass matrix $\bsf{M}$ is a diagonal matrix consisting of individual masses and is positive definite. Hence, it is easily invertible. The stiffness matrix $\bsf{S}$ is real symmetric and can be constructed with nonnegative\footnote{The spring constant of a nonexistent spring is zero.} spring constants $\kappa_{ij}$~\cite{RN121} .
Particularly, one can let the spring constants $\kappa_{ij}$ ($i\neq j$) be the ones connecting mass $i$ and mass $j$; while the spring constants $\kappa_{ii}$ are the ones connecting mass $i$ to a fixed end like a wall~\cite{gladwell_book,babbush2023exponential}. This is also known as a fixed-fixed system~\cite[Chap.~4]{gladwell_book} and the stiffness matrix $\bsf{S}$ is also positive definite\footnote{If there is no wall, $\bsf{S}$ can be positive semidefinite. Rigid body mode exists due to the system's freedom to move as a whole.}. 
% \cite{*[{See Chapter 4 of: }] [{}] gladwell_book_2}
The vector $\pmb{x}(t)$ specifies the displacement of each mass from its equilibrium position. Eq.~\ref{eq:2ndorderODE_spring} is the result of Newton's second law of motions for $2N$ masses~\cite{babbush2023exponential}. It can also be a broader result of Hamilton's principle (see Appendix~\ref{app:spring_mass_equation_derivation} for derivation), which can be generalized to other physical phenomena such as neural network~\cite{tsang1991dynamics,stiefel2016neurons, gerstner2014neuronal} and superconducting arrays~\cite{brown2003globally}.

Looking at Eq.~\ref{eq:2ndorderODE_spring}, there are some degrees of freedom in how $\bsf{S}$ can be expressed without changing the dynamics of the spring-mass system. While $\pmb{x}(t)$ specifies the displacement of each mass from its equilibrium position, the positive axis direction of individual mass can be assigned differently. Below we explain two encoding schemes: convention encoding, where every mass shares the same positive axis direction; and unconventional encoding, where some of the masses have their positive displacements defined opposite to others. At last, we introduce a special trick to embedd a Hamiltonian into the spring-mass system: translation and truncation.

\subsubsection{Conventional encoding: every mass shares the same positive axis direction}
\label{sec:conventional}
In the conventional encoding scheme, every mass shares the same positive axis direction. For example, for a 1-D system, the positives are either all to the right or all to the left. In this way, the off-diagonal element of $\bsf{S}$ is the negative of spring constant, and the diagonal elements are the sum of the spring constants connecting to the respective mass: $\bsf{S}_{ij} = -\kappa_{ij}$ for $i \neq j$
and
$\bsf{S}_{ii} = \sum_{j}\kappa_{ij}$. 
If $\bsf{S}$ is constructed in this way, it is a Z-matrix with nonnegative diagonals. We here introduce the notion of a L$_{0}$-matrix~\cite{azimzadeh2019fast}.

\begin{definition}
A L$_{0}$-matrix is a Z-matrix with nonnegative diagonal entries.
\end{definition}

In this conventional encoding scheme, $\bsf{S}$ is a L$_{0}$-matrix. Since for every row $i$, $\bsf{S}_{ii} \geq \sum_{j\neq i}|\bsf{S}_{ij}|$, it is also weakly diagonally dominant~\cite{zhao2018new,azimzadeh2019fast}.
\begin{figure*}[htb!]  
  \subfigure[]{
  \centering
\includegraphics[width=0.99\columnwidth]{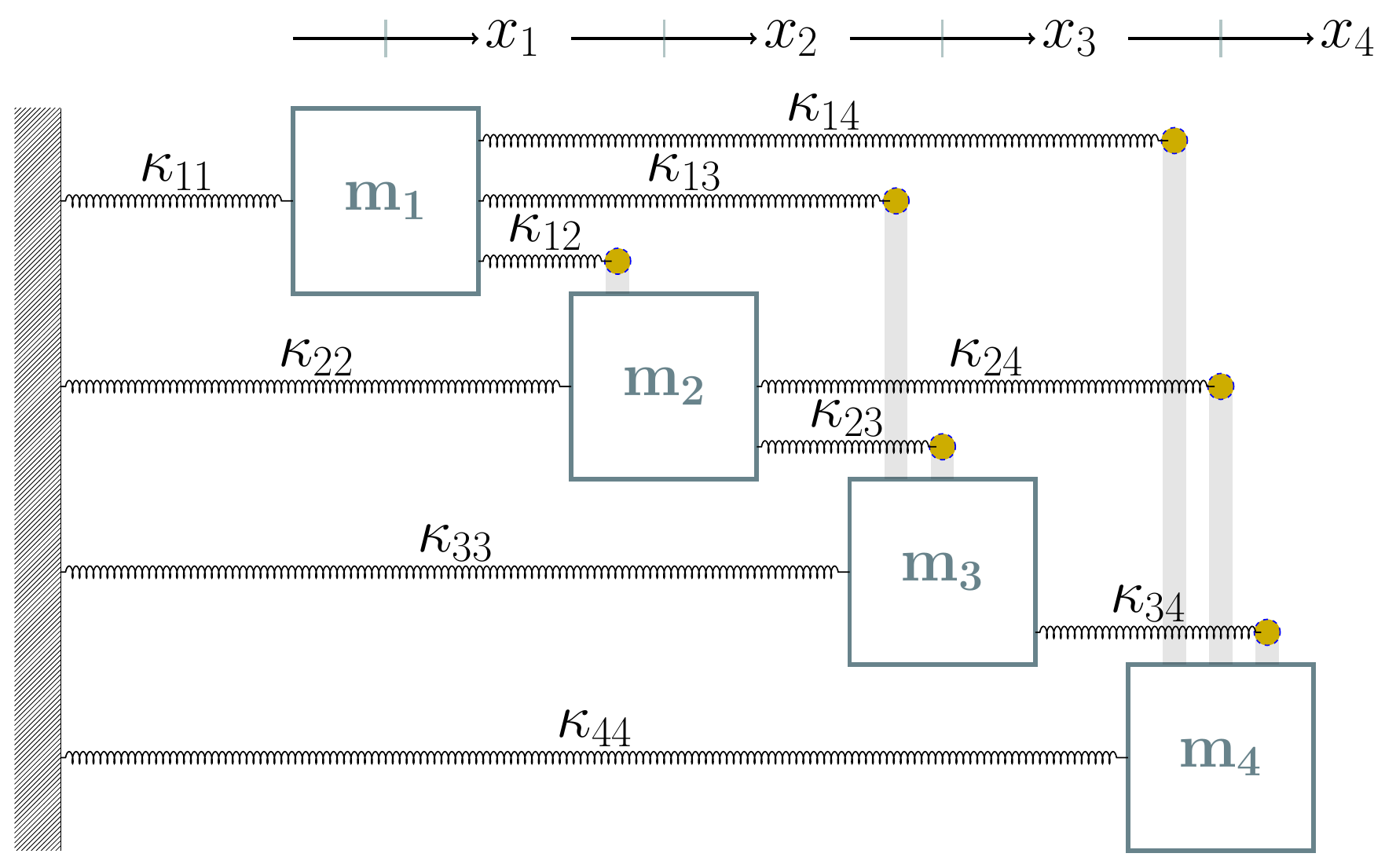}  
    \label{fig:massspring_12_a}
    }
    \hspace{0.1mm}
  \subfigure[]{
    \centering
\includegraphics[width=0.99\columnwidth]{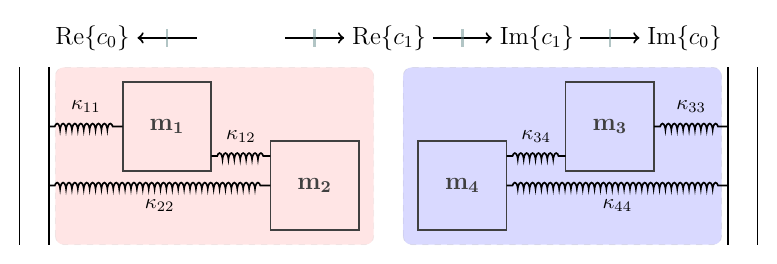} 
    \label{fig:massspring_12_b}
    }
  \caption{(a) An illustration of a fully-connected one-dimensional spring-mass system with $4$ masses and $10$ springs. The spring constants are labeled explicitly. 
    In this crowded scenario, the spring connections are drawn such that they are to the rods above the masses. Only lateral movements are allowed. 
    % Such system can simulate the dynamics of $1$ qubit under the drive of any time-independent Hamiltonian. 
    On top we label the axis of each mass, where every mass has the same positive direction. (b) An equivalent but decoupled system, and some springs to attached to another wall. Physically, it does not matter which wall the springs are attached to; they are drawn this way only for clarity. This system can simulate the dynamics of any time-independent real symmetric one-qubit Hamiltonian, with $\ket{\psi(t)} = c_{0}(t)\ket{0}+c_{1}(t)\ket{1}$.  On top we label the axis of each mass, where mass $1$ has an opposite positive compared to others.} 
    \label{fig:massspring_12}
\end{figure*}

One can invert the positive definite matrix $\bsf{M}$ and redefine $\bsf{S}' \equiv \bsf{M}^{-1}\bsf{S}$ to form a new equation:
\begin{equation}
\label{eq:2ndorderODE_spring_new2}
    \ddot{\pmb{x}}(t) + \bsf{S}'\pmb{x}(t) = 0\,,
\end{equation}
which coincidentally equals Eq.~\ref{eq:2ndorderODE_spring_new} of the second-order linear dynamical system. 
However, note that such a spring-mass system has a positive and diagonal mass matrix $\bsf{M}$, making the redefinition of $\bsf{S}$ straightforward. 

The quantum simulation using the spring-mass system is then to replace the Schr\"odinger equation in real form (Eq.~\ref{eq:realform}) with Eq.~\ref{eq:2ndorderODE_spring_new2}. The conventional encoding of $\bcal{K}$ to $\bsf{S}$ is often difficult. However, it is feasible in some cases, such as when $\bsf{H}$ is a ZM matrix (which implies that $\bcal{K}$ is a L$_{0}$-matrix). Other examples include a Hamiltonian that composed only of $\sigma_i^{z}$, such as the Ising model, where the $\bcal{K}$ is strictly diagonal. Again, one need not worry about if there is any negative entry on the diagonals --- because, as shown in Sec.~\ref{sec:positivesemidefinite_Kcal} or Table.~\ref{tab:K_vs_stiffness_simple}, $\bcal{K}$ is always positive semidefinite and real symmetric --- which automatically implies nonnegative diagonal entries.

For general Hamiltonians, the signs of off-diagonals can be a problem\footnote{There is one way to solve this sign problem: by engineering negative spring constants. By using clever tricks one can build negative equivalent spring constants, for example, buckling~\cite{churchill2016dynamically}, via two extra springs and a hinge~\cite{zhao2016negative, wang2004extreme}, and magnetic springs~\cite{tu2020novel,oyelade2017dynamics}. Although these engineering tricks induce local instabilities within the system, the stiffness matrix of the system remains globally stable~\cite{wang2004extreme}. Since $\bcal{K}$ is positive semidefinite, the global stablilty is guaranteed. Nevertheless, we do not look into such engineering tactics; instead we let every spring constant be strictly positive.}. Therefore, we introduce an unconventional encoding scheme.

\subsubsection{Unconventional encoding and the change of signs}
\label{sec:unconventional}

The general observation is, if the default positive axis of mass $i$ is ``flipped", the signs of these particular off-diagonal matrix elements $\bsf{S}_{ij}$ have to switch as follows, to preserve the dynamics described in Eq.~\ref{eq:2ndorderODE_spring}.
\begin{align}
    \bsf{S}_{ij} \rightarrow (-1)\times \bsf{S}_{ij} \,\,\,\,\, \forall j \neq i \label{eq:flip1}\\
    \bsf{S}_{ji} \rightarrow (-1)\times \bsf{S}_{ji}
    \,\,\,\,\, \forall j \neq i \label{eq:flip2}
\end{align}

Note also that Eq.~\ref{eq:flip2} is simply a consequence of Eq.~\ref{eq:flip1} for a real symmetric $\bsf{S}$. Fig.~\ref{fig:signswitching} is a pictorial way to visualize such sign switching.

\begin{figure}[htb!]
    \centering
\includegraphics[width=0.25\textwidth]{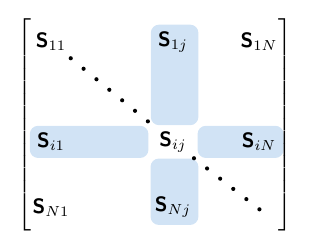}
    \caption{The highlighted matrix elements switch signs if the positive axis of mass $i$ flips. Other matrix elements of $\bsf{S}$ is unaffected.}
    \label{fig:signswitching}
\end{figure}

To provide further context, we present some examples of this unconventional encoding scheme:
\begin{enumerate}
    \item[Ex. 1:] Consider a sum of two transverse fields: $\bsf{H} = \sum_{i=1}^{2} \sigma_i^{x}$. Its dynamics can be simulated by a spring-mass system with $8$ masses and $4$ springs; with masses $3$, $4$, $7$, $8$ having an opposite positive-axis directions compared to the rest.
    \item[Ex. 2:] Consider a $3$-qubit Heisenberg XY chain~\cite{farreras2024simulation}:
$\bsf{H} = \sum_{i=1}^{2} \sigma_i^{x}\sigma_{i+1}^{x} + \sigma_i^{y}\sigma_{i+1}^{y}$. 
Its dynamics can be simulated by a spring-mass system with $16$ masses and $8$ springs; with masses $5$, $7$, $13$, $15$ having an opposite positive-axis directions compared to the rest.
\end{enumerate}

\subsubsection{Translation and truncation}
\label{sec:translate_and_truncate}
Consider again the conventional encoding scheme in Sec.~\ref{sec:conventional}, where every mass shares the same positive direction. Given a Hamiltonian $\bsf{H}$ of which the $\bcal{K}$ has many positive off-diagonals, 
one can define a translated Hamiltonian $\bsf{H}' = \bsf{H} - \alpha I$, from which a new $\bcal{K}'$ can be derived. Note that $\bsf{H}$ and $\bsf{H}'$ share the same expectation value dynamics. As $\alpha$ increases, the relative importance of the positive off-diagonals in $\bcal{K}'$ may decrease\footnote{For one thing, $\bcal{K}'$ becomes more diagonally dominant as $\alpha$ increase; for another, some positive off-diagonals of $\bcal{K}'$ might become negative.} and such entries may be truncated without affecting the dynamics. One can thus define:
\begin{align}
\label{eq:truncation}
    \bcal{K}'' &= \textsf{Trunc}(\bcal{K}')  \nonumber\\
    &= \textsf{Trunc}\left(
    \begin{bmatrix}[cc]
      \Re{(\bsf{H}')^2} & -\Im{(\bsf{H}')^2} \\
      \Im{(\bsf{H}')^2} & \Re{(\bsf{H}')^2}
    \end{bmatrix}
    \right) 
\end{align}
where the $\textsf{Trunc}(\cdot)$ operation is to remove the positive off-diagonals, i.e., set them to zero. One can then easily embed the $\bcal{K}''$ into a spring-mass setup\footnote{The assumption is that the positive semidefiniteness of $\bcal{K}''$ is preserved after the $\textsf{Trunc}(\cdot)$ operation.}. 

We illustrate such technique with an example of a sum of four transverse fields: 
\begin{equation}
\label{eq:sum4tf}
    \bsf{H} = \sum_{i=1}^{4}\sigma_i^{x} \,.
\end{equation}
With the initial state of $\ket{\psi(0)}=\ket{0000}$, we want to calculate the dynamics of $|c_{0}(t)|^2 = |\langle 0000|\psi(t)\rangle|^2$ from $t=0$ to $t=1$. The quantum evolution by Schr\"odinger equation (Eq.~\ref{eq:Schr}) is shown by the solid line in Fig.~\ref{fig:translatetruncation}.
\begin{figure}[htb!]
    \centering
\includegraphics[width=0.475\textwidth]{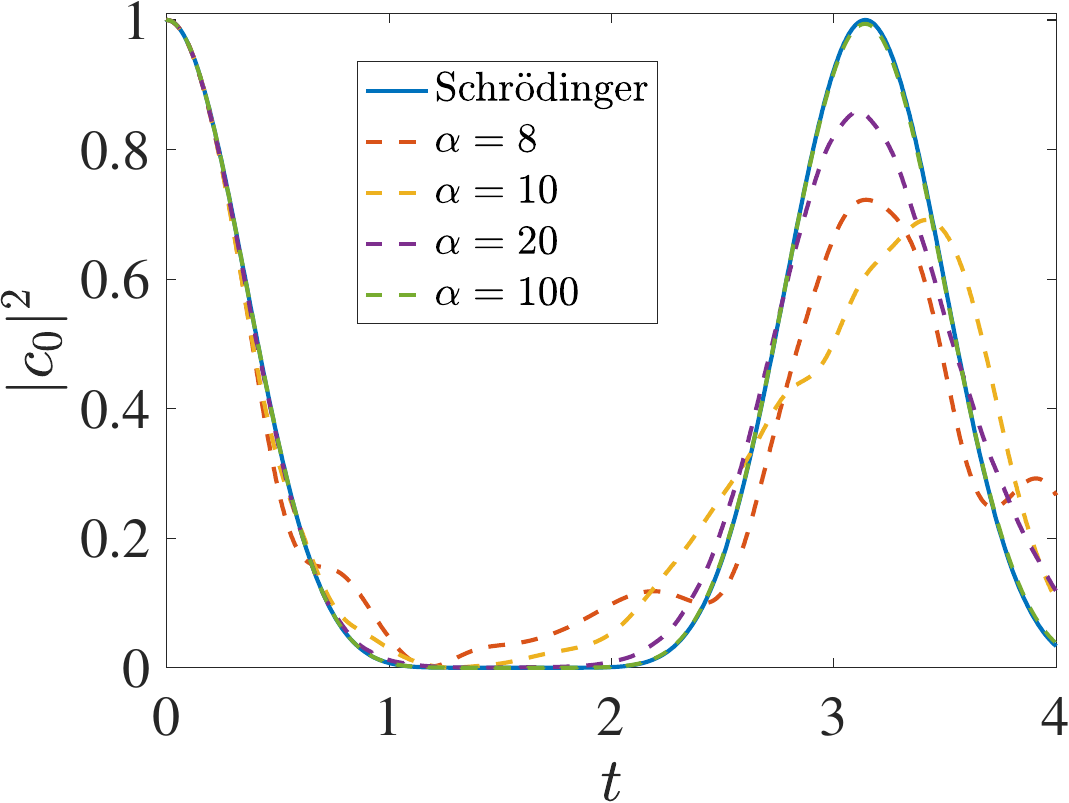}
    \caption{Solid line represents the $|c_0(t)|^2$ of $\ket{\psi(t)}$ under the Hamiltonian evolution of Eq.~\ref{eq:sum4tf}. Dotted line represents the $|c_0(t)|^2$ calculated from the displacements of masses during the spring-mass system evolution. Spring constants are derived from the technique of ``translation and truncation" with $4$ different values of $\alpha$. The approximation is almost perfect with $\alpha = 100$.}
    \label{fig:translatetruncation}
\end{figure}

The conventional encoding scheme fails because of the positive off-diagonals in $\bsf{H}^2$. For the simulation of spring-mass system, we evaluate different translated Hamiltonian $\bsf{H}' = \bsf{H} - \alpha I$, corresponding to different values of $\alpha \in \{8,10,20,100\}$. For each $\alpha$, we calculate $\bcal{K}''$ using the truncation operation in Eq.~\ref{eq:truncation}, then embed $\bcal{K}''$ into a spring-mass stiffness matrix directly as $\bsf{S} = \bcal{K}''$. For all $\alpha$, the initial displacements $\pmb{x}(0)$ of all masses are zero except mass $1$, where $\pmb{x}_{1}(0) = 1$. The initial velocities are calculated by Eq.~\ref{eq:ivp_real}, i.e., $-\begin{bmatrix}
\bsf{0} & \hspace{-0.8ex}-\Re{\bsf{H}'} \\
\Re{\bsf{H}'} & \hspace{-0.8ex} \bsf{0}
\end{bmatrix}\pmb{x}(0)$. The spring-mass system is evolved under the dynamics equation Eq.~\ref{eq:2ndorderODE_spring_new}. For all $\alpha$, we obtain $|c_0(t)|^2 = \pmb{x}_{1}(t)^2+\pmb{x}_{17}(t)^2$ and plot them as dotted lines in Fig.~\ref{fig:translatetruncation}. With a larger $\alpha$, one can approximate the quantum evolution by the spring-mass evolution more accurately. The numerical recipe for the classical dynamics simulation is outlined in Appendix.~\ref{app:numericalrecipe}.

\section{Analog classical speedup with spring-mass system}
\label{sec:speedup}
In this section, we study a spring-mass system with scaled parameters, aiming to achieve a speedup relative to quantum evolution. For clarity, all quantities are presented with their respective physical units.
We consider the following $4$-qubit Hamiltonian: 
\begin{equation}
\label{eq:4-qubit-speedup}
    \bsf{H} = g\sigma_1^{x} + J\sum_{i=1}^{3} \sigma_i^{z}\sigma_{i+1}^{z}  \,,
\end{equation}
where $g=0.5\,$GHz and $J=2\,$GHz, which are of the similar order of magnitude as the transverse field and coupling strengths in modern quantum electronics. The initial state is an equal superposition state  $\ket{\psi(0)}= \ket{+}^{\otimes 4}$, where $\ket{+} = \frac{\ket{0}+\ket{1}}{\sqrt{2}}$. 

To simulate the quantum dynamics of this Hamiltonian, one can build a device with $32$ masses and $40$ springs, where $8$ springs connect the masses to another, and the remaining $32$ springs attach the masses to a fixed wall. We use the unconventional encoding scheme in Sec.~\ref{sec:unconventional} with masses $1$, $16$, $17$, $32$ having an opposite positive-axis direction compared to the other masses. One can measure the observables in SI unit:  displacement of the masses in meter (m), mass of the masses in kilogram (kg), and the spring constant in Newton per meter (N/m). In our example, each mass has a mass of $1$kg. In general, the smaller the mass, the faster the dynamics  (See, e.g., Eq.~\ref{eq:2ndorderODE_spring_new2}). 

To emulate the uniform superposition initial state, the initial displacements $\pmb{x}(0)$ of the masses from their equilibrium points are:
\begin{align}
\text{For all $i\in \{1, \dots,16\}$:      } \pmb{x}_{i}(0)  &= +1.25\text{\,m}  \,; \nonumber\\
\text{For all $i\in \{17, \dots,32\}$:      } \pmb{x}_{i}(0)  &= 0\,,
\label{eq:initialdisplacements}
\end{align}
Note that the size of the spring-mass system can be arbitrarily contracted or enlarged\footnote{The size depends on the constant $d$ elaborated under Eq.~\ref{eq:2ndorderODE_spring_new}, which can be as small or as large as desired, provided the system does not approach the quantum or relativistic limits.} without affecting the dynamics timescale. Therefore, the exact order of magnitude in Eq.~\ref{eq:initialdisplacements} is not of importance.

\begin{figure*}[htb!]  
\label{fig:classicalspeedup}
  \subfigure[]{
  % \centering
\includegraphics[width=0.9\columnwidth]{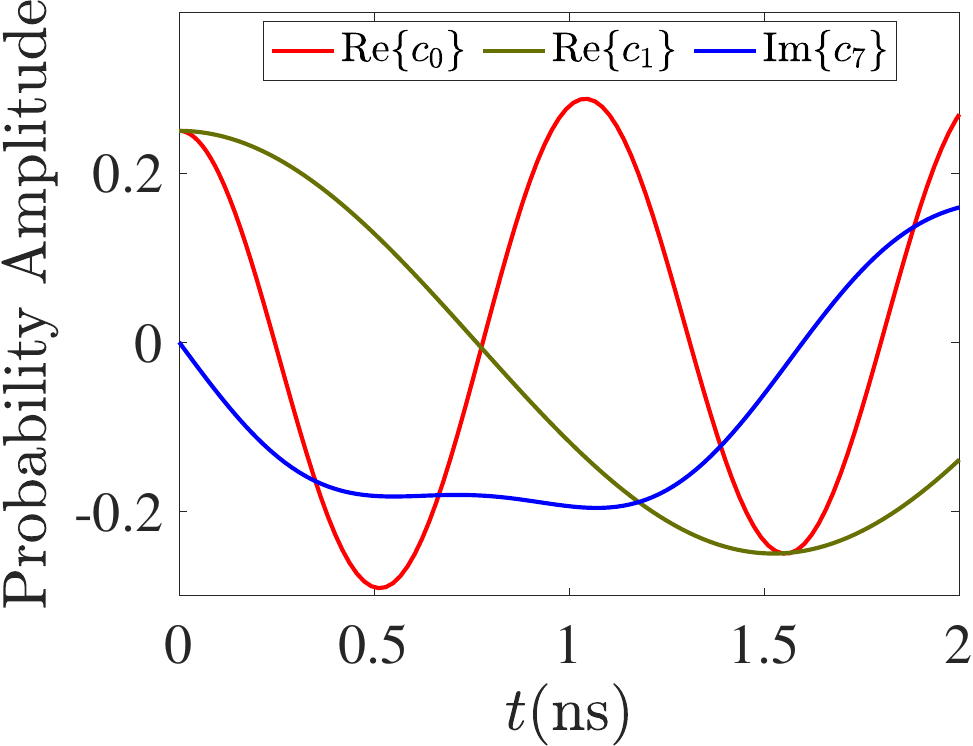}  
\label{fig:classicalspeedup_1}
    }
    \hspace{10mm}
  \subfigure[]{
    % \centering
\includegraphics[width=0.9\columnwidth]{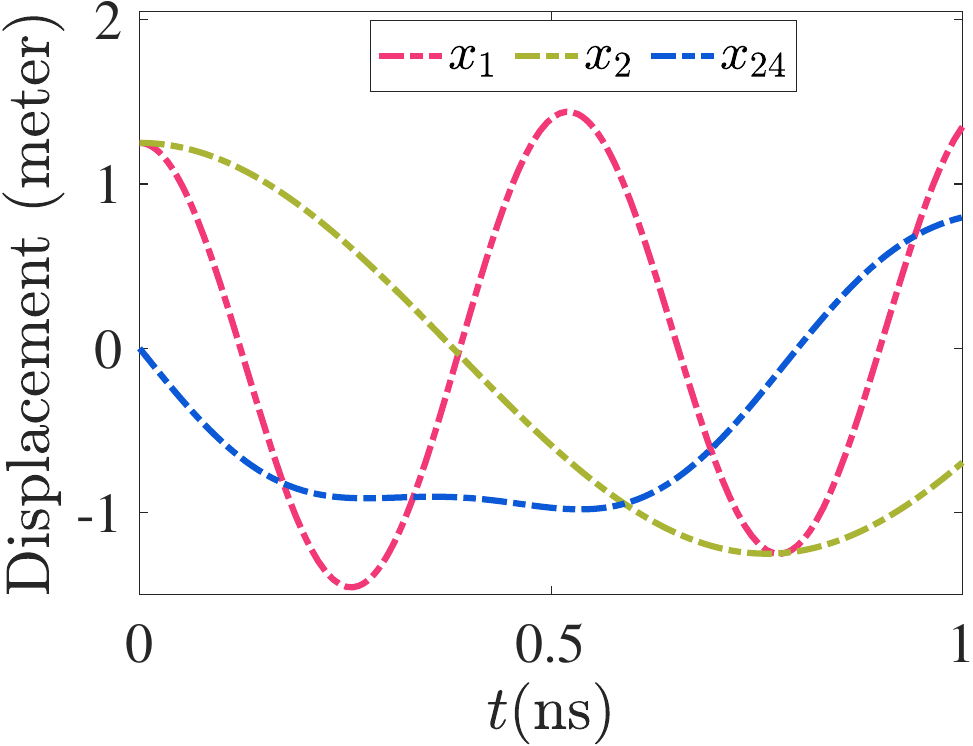} 
    \label{fig:classicalspeedup_2}
    }
  \caption{(a) Evolution of $\Re{c_{0}(t)}$, $\Re{c_{1}(t)}$ and $\Im{c_{7}(t)}$ under the drive of the Hamiltonian Eq.~\ref{eq:4-qubit-speedup} for a duration of $2\,$ns. (b) Displacement of the masses $1$, $2$ and $24$ relative to their equilibriums, under the drive of spring constants Eq.~\ref{eq:4-qubit-speedup_springs} for a duration of $1\,$ns. The two graphs are identical, differing only in the physical units on the y-axis and the timescale on the x-axis.} 
\end{figure*}

Given the initial displacements, it is straightforward to deduce from Eq.~\ref{eq:ivp_real} that an $p$-time linear speedup can be achieved by scaling the spring constants by $p^2$ and the initial velocities by $p$. In this example, we aim to achieve a $2$-time speedup compared to the quantum evolution timescale, so we set $p=2$. Specifically, the spring constants are as follows:
\begin{align}
\label{eq:4-qubit-speedup_springs}
\kappa_{1,9} = \kappa_{3,11} = \kappa_{6,14} = \kappa_{8,16} &= 16\text{\,GN/m}\,, \nonumber\\
\kappa_{17,25} = \kappa_{19,27} = \kappa_{22,30} = \kappa_{24,32} &= 16\text{\,GN/m}\,, \nonumber\\
\kappa_{2,2} = \kappa_{4,4} = \kappa_{5,5} = \kappa_{7,7} &= 17\text{\,GN/m}\,,\nonumber\\
\kappa_{10,10} =  \kappa_{12,12} = \kappa_{13,13} = \kappa_{15,15} &= 17\text{\,GN/m}\,,\nonumber\\
\kappa_{18,18} = \kappa_{20,20} = \kappa_{21,21} = \kappa_{23,23} &= 17\text{\,GN/m}\,,\nonumber\\
\kappa_{26,26} =  \kappa_{28,28} = \kappa_{29,29} = \kappa_{31,31} &= 17\text{\,GN/m}\,,\nonumber\\ 
\kappa_{1,1} =  \kappa_{6,6} = \kappa_{11,11} = \kappa_{16,16} &= 129\text{\,GN/m}\,,\nonumber\\ 
\kappa_{17,17} =  \kappa_{22,22} = \kappa_{27,27} = \kappa_{32,32} &= 129\text{\,GN/m}\,,\nonumber\\ 
\kappa_{3,3} =  \kappa_{8,8} = \kappa_{9,9} = \kappa_{14,14} &= 1\text{\,GN/m}\,,\nonumber\\ 
\kappa_{19,19} =  \kappa_{24,24} = \kappa_{25,25} = \kappa_{30,30} &= 1\text{\,GN/m}\,.
\end{align}
Note that the spring constant depends on specific spring material and size, and theoretically can be made arbitrarily large by connecting many springs in parallel\footnote{Recall the equivalent spring constant formula for parallel springs.}. For additional information, the initial velocities $\pmb{v}(0) = \dot{\pmb{x}}(0)$ of the masses are: 
\begin{align}
\text{For all $i\in \{1, \dots,16\}$:      } \pmb{v}_{i}(0)  &= 0  \,, \nonumber\\
\pmb{v}_{17}(0) = \pmb{v}_{32}(0)  &= -3.25\text{\,m\,s}^{-1}\,, \nonumber\\
\pmb{v}_{18}(0) = \pmb{v}_{20}(0)  = \pmb{v}_{24}(0) &= -1.25\text{\,m\,s}^{-1}\,, \nonumber \\
\pmb{v}_{25}(0) = \pmb{v}_{29}(0) = \pmb{v}_{31}(0) &= -1.25\text{\,m\,s}^{-1}\,, \nonumber \\
\pmb{v}_{19}(0) =  \pmb{v}_{21}(0) = \pmb{v}_{23}(0) &= +0.75\text{\,m\,s}^{-1}\,, \nonumber \\
\pmb{v}_{26}(0) = \pmb{v}_{28}(0) = \pmb{v}_{30}(0) &= +0.75\text{\,m\,s}^{-1}\,.
\end{align}

We are interested in the dynamics of probability amplitudes: $\Re{c_{0}(t)}$, $\Re{c_{1}(t)}$ and $\Im{c_{7}(t)}$, where in big-endian ordering, $c_{0}(t) = \langle 0000|\psi(t)\rangle$, $c_{1}(t) = \langle 0001|\psi(t)\rangle$ and $c_{7}(t) = \langle 0111|\psi(t)\rangle$.  Fig.~\ref{fig:classicalspeedup_1} shows the time evolution of  $\Re{c_{0}(t)}$, $\Re{c_{1}(t)}$ and $\Im{c_{7}(t)}$ in a closed quantum system. The values are obtained by solving the Schr\"odinger equation with the given values of transverse field and coupling strengths. Fig.~\ref{fig:classicalspeedup_2} shows the displacements $\pmb{x}_1(t)$, $\pmb{x}_2(t)$ and $\pmb{x}_{24}(t)$ of the masses with spring constants tuned to the values of Eq.~\ref{eq:4-qubit-speedup_springs}. The displacement data are obtained by integrating the spring-mass dynamics equation Eq.~\ref{eq:2ndorderODE_spring}. The results suggest that the spring-mass system evolves twice as fast as the quantum system in real-world clock time. 

%%%%%%%%%%%%%%%%%%%%%%%%%%%%%%%%%%%%%%%%%%%%%%%
\section{Analog classical simulation of an Eight-qubit QAOA (Max-cut)}
\label{sec:qaoa}
%%%%%%%%%%%%%%%%%%%%%%%%%%%%%%%%%%%%%%%%%%%%%%%
In this section, we explore the possibility of using an analog computer to execute a quantum algorithm known as the Quantum Approximate Optimization Algorithm (QAOA)~\cite{FarhiQAOA}. While the analog simulation is also possible via a spring-mass system with some supplementary encoding techniques, we move back to the General-Purpose Analog Computer (GPAC) and use the notations from Algorithm~\ref{alg:sim1} for a more general illustration.  

Consider a Max-Cut problem for a cycle graph of eight vertices and eight edges connecting each pair of neighboring vertices. Let $\bsf{H}_C$ represent the cost Hamiltonian and $\bsf{H}_B$ denote the mixer Hamiltonian. Specifically,
\begin{align}
    \bsf{H}_C &= \frac{1}{2} \sum_{(i,j) \in E} (1-\sigma_i^{z}\sigma_j^{z}) \label{eq:cost_op}\,,\\
    \bsf{H}_B &= \sum_{i} \sigma_i^{x} \label{eq:mixer_op}\,,
\end{align}
with $\bsf{H}_C$ encoding an eight-vertex graph with the edge set $E = \left\{(1,2),(2,3),(3,4),(4,5),(5,6),(6,7),(7,8),(8,1)\right\}$, 
and $\bsf{H}_B$ being the sum of tranverse fields. 

The variational quantum state is prepared by applying a unitary transformation composed of $p$ blocks of phase-separating unitaries 
$\left\{U_{\mathrm{C}}(\gamma_k) = e^{-i \gamma_k \bsf{H}_{\mathrm{C}}} \right\}_{k=1}^p$ 
and mixing unitaries 
$\left\{U_{\mathrm{B}}(\beta_k) = e^{-i \beta_k \bsf{H}_{\mathrm{B}}} \right\}_{k=1}^p$ to an initial state, such that the state after the $p$ blocks is given by
\begin{equation}
\label{eq:p_blocks}
\ket{\psi_p(\vec{\gamma},\vec{\beta})} = \prod_{k=1}^p U_{\mathrm{B}}(\beta_k) U_{\mathrm{C}}(\gamma_k) \ket{\psi_0} \,.
\end{equation}
The phases of these unitaries, $\left\{\beta_k, \gamma_k\right\}_{k=1}^p$, are the variational parameters to be optimized based on the quantum state outputs of the circuit. 

Each unitary is equivalent to integrating the Schr\"{o}dinger equation for one unit of time. Therefore, the QAOA algorithm can be executed on the GPAC by applying Algorithm~\ref{alg:sim1} repeatedly, specifically $2p$ times. Instead of implementing the $2p$ unitaries with quantum gates in a quantum circuit, one can set up a GPAC where the internal components are calibrated with the phases:
\begin{align}
\bsf{K}_{\gamma_k,C} &= (\gamma_{k}\bsf{H}_C)^2 \,,\\
 \bsf{K}_{\beta_k,B} &= (\beta_{k}\bsf{H}_B)^2 \,. 
\end{align}

Therefore, one can wrap $2p$ sequences of  Algorithm~\ref{alg:sim1} in a program to simulate the final output of 
$\left\{\bsf{K}_{\gamma_k,C}\rightarrow\bsf{K}_{\beta_k,B}\right\}_{k=1}^{p}$. The initial state of each sequence is the output of the previous one. The initial first derivatives $\ket{\dot{\varphi}(0)}$ of each sequence is calculated by the initial real state $\ket{\varphi(0)}$ and the corresponding calibrated Hamiltonian.
Compared to gate-based circuit, the real state $\ket{\varphi(t)}$ can be monitored in real time without state collapse during the GPAC operation. The evolution of objective value (i.e., number of cuts in the graph) can be computed efficiently as
\begin{equation}
\langle \varphi(t)| \bsf{H}_C \oplus \bsf{H}_C|\varphi(t)\rangle \,,
\end{equation}
with $\bsf{H}_C \oplus \bsf{H}_C$ being a real diagonal matrix.

For our example, the initial state $\ket{\psi_0}$ is the uniform superposition  state $\ket{+}^{\otimes 8}$ and $p=2$. Therefore, Eq.~\ref{eq:p_blocks} becomes
\begin{equation}
\label{eq:2_blocks}
\ket{\psi_2(\gamma_2,\gamma_1,\beta_2,\beta_1)} = \prod_{k=1}^2 U_{\mathrm{B}}(\beta_k) U_{\mathrm{C}}(\gamma_k) \ket{+}^{\otimes 8} \,.
\end{equation} 
We want to find the maximum expectation value of $\bsf{H}_C$ with $\ket{\psi_2(\gamma_2,\gamma_1,\beta_2,\beta_1)}$, among all possible phases. The range of phases to search for are $\beta_1 \in \left[-1, 1\right], \beta_2 \in \left[-1, 1\right], \gamma_1 \in \left[-1, 1\right]$ and $\gamma_2 \in \left[-1, 1\right]$. The optimal solution  is $\{\gamma_2^\ast=0.5,\gamma_1^\ast=-1,\beta_2^\ast=0.8,\beta_1^\ast=0.4\}$, with 
$\langle \psi_2(\vec{\gamma}^\ast,\vec{\beta}^\ast)|\bsf{H}_C|\psi_2(\vec{\gamma}^\ast,\vec{\beta}^\ast)\rangle \approx 6.6$, where the expectation value is optimal. 

We numerically demonstrate an analog classical simulation scheme comprising four repetitions of Algorithm~\ref{alg:sim1}, each lasting one unit of time, to solve the QAOA problem in an alternative and equivalent manner. Fig.~\ref{fig:realandfirstderivative} shows the time evolution of the first matrix element of $\ket{\varphi (t)}$ and the first matrix element of $\ket{\dot{\varphi} (t)}$ during the four sequences of Algorithm~\ref{alg:sim1}. Note that each sequence depends on its phase and we focus on the phases of $\gamma_2=0.5,\gamma_1=-1,\beta_2=0.8,\beta_1=0.4$, which are also the optimal phases. 
While $\ket{\varphi (t)}$ remains continuous during the GPAC sequences, $\ket{\dot{\varphi} (t)}$ is reset at the beginning of each sequence.

\begin{figure}
    \centering
\includegraphics[width=0.5\textwidth]{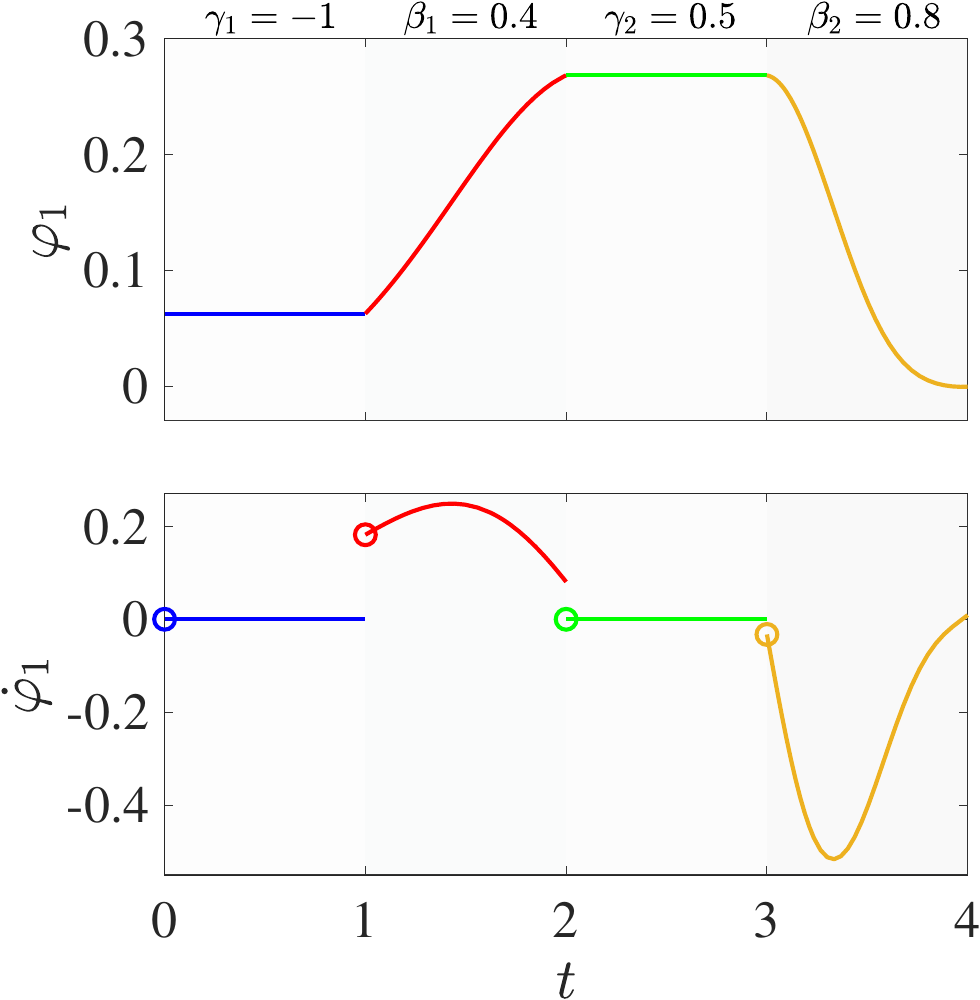}
    \caption{
    The time evolution of the first matrix element of $\ket{\varphi (t)}$ and the first matrix element of $\ket{\dot{\varphi} (t)}$ during the four sequences of Algorithm~\ref{alg:sim1}, for the phases of 
$\gamma_2^\ast=0.5,\gamma_1^\ast=-1,\beta_2^\ast=0.8,\beta_1^\ast=0.4$. While the initial state of each sequence is the output of the previous seqeuence, the initial first derivative is reset along with the inner workings of Algorithm~\ref{alg:sim1}.} 
\label{fig:realandfirstderivative}
\end{figure}

As shown in Fig.~\ref{fig:expectation_value_evolution}, one can reveal the time evolution of the expectation value throughout the GPAC operation (from start to finish). We again focus on the optimal solution. Fig.~\ref{fig:heatmap} shows the heatmap of $\langle \varphi(4)| \bsf{H}_C \oplus \bsf{H}_C|\varphi(4)$ at the end of the GPAC operation, for every possible phases of $\gamma_1 \in [-1,1], \beta_1 \in [-1,1]$ with $\gamma_2 = 0.5, \beta_2 = 0.8$. This heatmap generated from the GPAC is exactly the same as the heatmap of expectation values of $\bsf{H}_C$ from the quantum unitary sequence in Eq.~\ref{eq:2_blocks} (not shown). 

\begin{figure*}[htb!]  
\label{fig:}
  \subfigure[]{
  % \centering
\includegraphics[width=0.939\columnwidth]{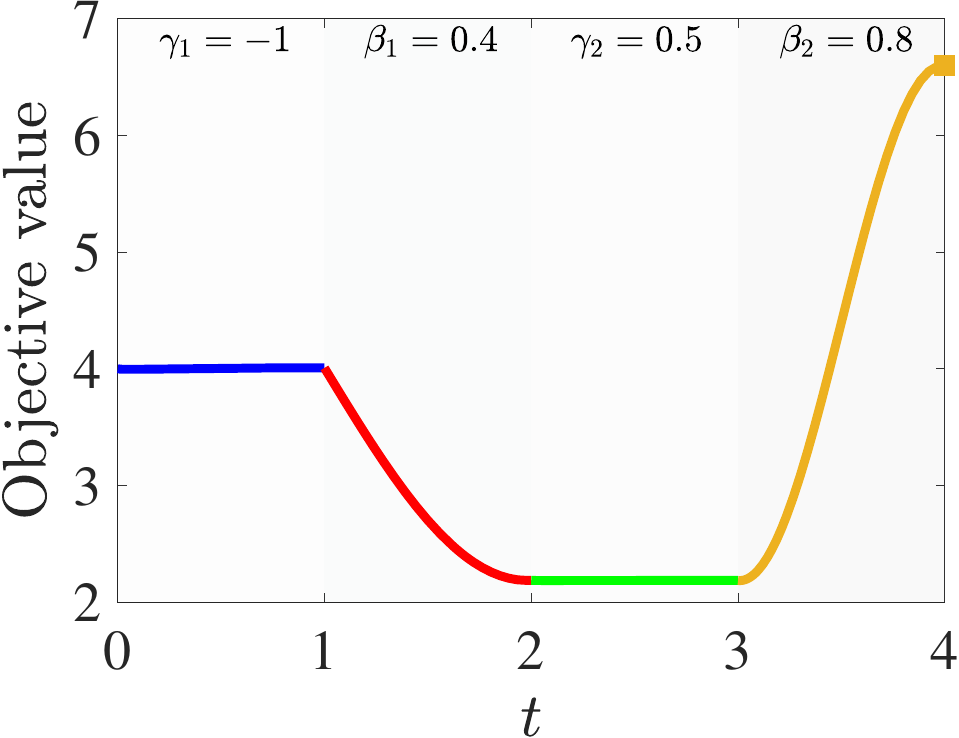}  
\label{fig:expectation_value_evolution}
    }
    \hspace{5mm}
  \subfigure[]{
    % \centering
\includegraphics[width=0.939\columnwidth]{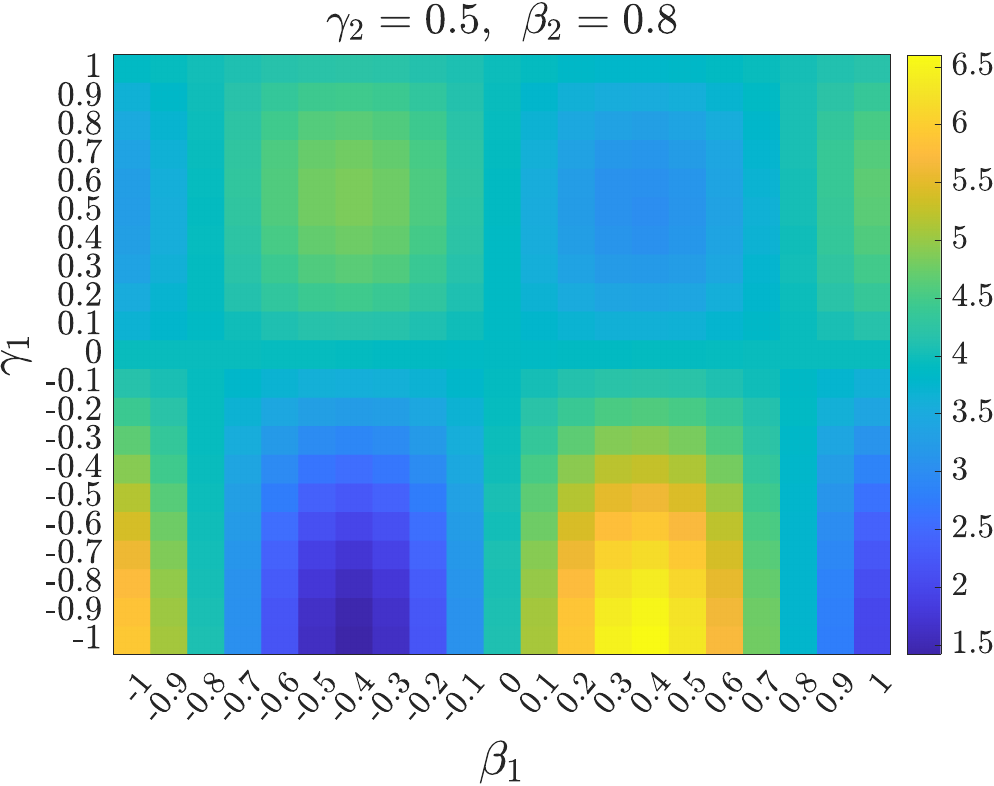} 
    \label{fig:heatmap}
    }
  \caption{(a) The time evolution of objective value $\langle \varphi(t)| \bsf{H}_C \oplus \bsf{H}_C|\varphi(t)\rangle$ during the four sequences of Algorithm~\ref{alg:sim1}, for the phases of $\gamma_2^\ast=0.5,\gamma_1^\ast=-1,\beta_2^\ast=0.8,\beta_1^\ast=0.4$. At the end of the last sequence, the objective value goes up to around $6.6$. (b) The heatmap of objective values $\langle \varphi(4)| \bsf{H}_C \oplus \bsf{H}_C|\varphi(4)\rangle$ at the end of the GPAC operations, for every possible phases of $\gamma_1 \in [-1,1], \beta_1 \in [-1,1]$ with $\gamma_2^\ast = 0.5, \beta_2^\ast = 0.8$.} 
\end{figure*}

\section{Discussions and Conclusions}
\label{sec:complexity}
The analog classical simulation algorithm requires the knowledge of the $\bsf{K}(t)$, to set up the rules of evolution of the simulator. For a time-independent Hamiltonian\footnote{For a time-dependent Hamiltonian, this requires also the knowledge of the first derivative.}, this requires the knowledge of the square of Hamiltonian $\bsf{K} = \bsf{H}^2$. There are several possible scenarios in which such knowledge can be obtained. The most probable one is that the analytical expressions of the Hamiltonian is known, in which case deriving an analytical form of  $\bsf{H}^2$ should not be difficult. 
Another possibility is that the Hamiltonian is given as a sum of Pauli strings, in which case the complexity of computing  $\bsf{H}^2$ is polynomial to the number of terms. The final possibility, which is rare, is that the Hamiltonian is given as a numerical matrix. In that case, we can only square it with brute force. 

We here compare the numerical complexity of squaring a Hamiltonian $\bsf{H}^2$ vs exponentiation $e^{\bsf{H}}$, as a streamlined approach to compare the efficiency of analog classical simulation vs digital classical simulation\footnote{Rigoriously one should compare $\bsf{H}^2$ vs $e^{-i\bsf{H}}$.}. Note that an analog computer can also perform basic arithmetic operations~\cite{blanc2023simulation}, but we do not consider such a case. We benchmark the comparison with the following Hamiltonian.
\begin{equation}
    \bsf{H} = \sum_{i=1}^{n}\sigma_i^x + \sum_{i=1}^{n}\sigma_i^z \,.
\end{equation}
We compare the brute-force computation time ratio $r$ of matrix exponential vs squaring, assuming we do not know the analytical expressions of these operations.
We compute the numerical matrix exponential using a \textsf{fastexpm}~\cite{fastexpm} function. The numerical evaluations were performed on a 10-core Apple M1 Pro chip with a 16 GB memory. Fig.~\ref{fig:timing} shows that the computation ratio for $n = \{12,13,14,15,16\}$. While the theoretical complexity of both the brute-force matrix exponentiations and matrix squarings should be exponential in $n$, the timing data in Fig.~\ref{fig:timing} seems to suggest that for practical calculations on sparse matrices, matrix squaring is much more computationally efficient, and scale more efficiently in $n$.

\begin{figure}[htb!]
    \centering
    \includegraphics[width=0.77\columnwidth]{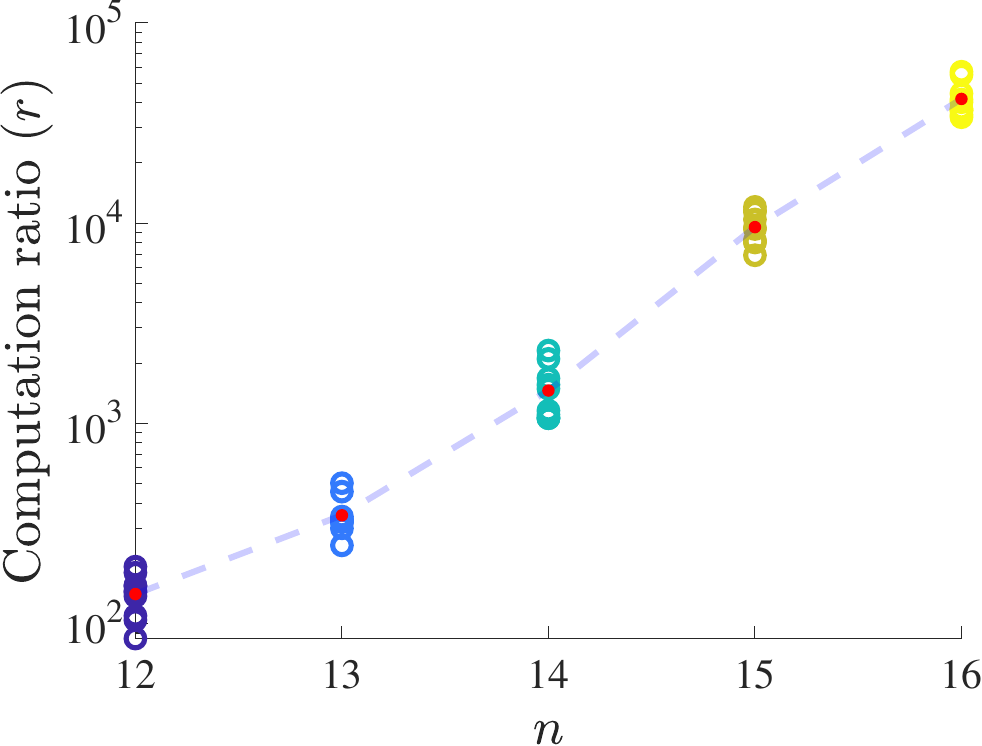}
    \caption{Computation ratio $r$ here refers to the computation time of sparse matrix exponentiation vs sparse matrix squaring. The y-axis is in log scale. Due to random execution background factors, we performed 10 repeated trials for each number of qubits $n$, with the mean values shown as red dots.}
    \label{fig:timing}
\end{figure}

In this work, we have shown how a closed quantum system can be simulated by an analog computer, which in the mathematical form, is a GPAC. The most important steps are perhaps mapping the Schr\"{o}dinger equation to a system of real ODEs, and how to map the real ODEs to classical physical features such as a stiffness matrix. While we outlined how to do the above steps and how to encode and realize such simulation in a second-order linear dynamical system and a spring-mass system, there exist multiple analogies of the dynamical equations (see Appendix~\ref{app:physics}, or Ref.~\cite{hasler2016opportunities}), shared by electrical and biological systems.
Many have related quantum computing to analog computing~\cite{laughlin2005different,ferry2003we, kish2003hilbert}. This is the first time to systematically study how to use the analog computers to simulate a general, including time-dependent, quantum dynamics.

Analog computers has been deemphasized since the Church-Turing thesis~\cite{bournez2013turing} and the success of digital computers. We hope that this work will motivate further research into this nearly obsolete computational model and, more broadly, into physical computing based on real variables for the exploration of quantum phenomena. One goal of quantum simulation is to simulate a Hamiltonian with another one in a well-controlled system. The analog devices might offer a more stable and robust-to-noise approach, thus allowing for a simulation longer than the decoherence timescale of quantum computers. The broad energy range of the analog devices may also offer an opportunity for the fast forwarding of quantum dynamics. 

%%%%%%%%%%%%%%%%%%%%%%%%%%%%%%%%%%%%%%%%%%%%%%%
\section{Acknowledgement}
%%%%%%%%%%%%%%%%%%%%%%%%%%%%%%%%%%%%%%%%%%%%%%%
We gratefully thank Dr. Daniel T. Kawano for discussions. This research received no specific grant from any funding agency in the public, commercial, or not-for-profit sectors. The numerical experiments were carried out on a personal laptop using codes in \url{https://github.com/kwyip/acs}.

\bibliography{refs}

\onecolumngrid
\appendix
\section{Encoding probability amplitudes as velocity: First-order ODE and damping}
\label{app:firstorder}
To convert the Schr\"{o}dinger equation into a system of first-order real-valued ODEs
\begin{align}
\label{Eq:Schr_app}
 \ket{\dot{\psi}(t)} &= -i \bm{\mathsf{H}}(t) \sket{\psi(t)} \,,
\end{align}

Let $\bsf{J}(t)$ denote $i\bsf{H}(t)$ such that
\begin{equation}
\label{Eq:Schr_real1st}
 \ket{\dot{\psi}(t)} = - \bsf{J}(t) \sket{\psi(t)} \,.
\end{equation}
Eq.~\ref{Eq:Schr_real1st} can be turned into a real form of
\begin{equation}
\label{Eq:Schr_real1st_2}
 \ket{\dot{\varphi}(t)} = - \bcal{J}(t) \sket{\varphi(t)} \,,
\end{equation}
where $\sket{\varphi(t)} = \begin{bmatrix}
    \Re{\ket{\psi(t)}} \\
    \Im{\ket{\psi(t)}}
    \end{bmatrix}$ and $\bcal{J}(t) = \begin{bmatrix}
  \Re{\bsf{J}(t)} & -\Im{\bsf{J}(t)} \\
  \Im{\bsf{J}(t)} & \Re{\bsf{J}(t)}
\end{bmatrix} = 
\begin{bmatrix}
  -\Im{\bsf{H}(t)} & -\Re{\bsf{H}(t)} \\
  \Re{\bsf{H}(t)} & -\Im{\bsf{H}(t)}
\end{bmatrix} 
$.

\begin{proof}
\begin{align}
     \ket{\dot{\varphi}(t)} = \frac{d}{dt}\begin{bmatrix}
    \Re{\ket{\psi(t)}} \\
    \Im{\ket{\psi(t)}}
    \end{bmatrix}
    &= \begin{bmatrix}
    \Re{\ket{\dot{\psi}(t)}} \\
    \Im{\ket{\dot{\psi}(t)}}
    \end{bmatrix} \\
    &= \begin{bmatrix}
    \Re{-i\bsf{H}(t)\sket{\psi(t)}} \\
    \Im{-i\bsf{H}(t)\sket{\psi(t)}}
    \end{bmatrix} \\
    &= \begin{bmatrix}
    \Re{-i (\Re{\bsf{H}(t)} + i\Im{\bsf{H}(t)})(\Re{\sket{\psi(t)}}+i\Im{\sket{\psi(t)}}) } \\
    \Im{-i (\Re{\bsf{H}(t)} + i\Im{\bsf{H}(t)})(\Re{\sket{\psi(t)}}+i\Im{\sket{\psi(t)}}) }
    \end{bmatrix} \\
    &= \begin{bmatrix}
    \Im{\bsf{H}(t)}\Re{\sket{\psi(t)}} + \Re{\bsf{H}(t)}\Im{\sket{\psi(t)}} \\
    -\Re{\bsf{H}(t)}\Re{\sket{\psi(t)}}+\Im{\bsf{H}(t)}\Im{\sket{\psi(t)}}
    \end{bmatrix} \\
    &= \left[\begin{array}{cc}
\Im{\bmsf{H}(t)} & \Re{\bmsf{H}(t)} \\
-\Re{\bmsf{H}(t)} & \Im{\bmsf{H}(t)}
\end{array}\right]\left[\begin{array}{c}
\Re{\ket{\psi(t)}} \\
\Im{\ket{\psi(t)}}
\end{array}\right] \\
&= - \bcal{J}(t) \sket{\varphi(t)}
\end{align}
\end{proof}

If $\bsf{H}(t)$ is a real Hamiltonian, $\bcal{J}(t)$ is antisymmetric. For a real Hamiltonian, it might be possible to simulate the dynamics with artificial systems, such as an antisymmetric recurrent neural network~\cite{antisymmetricrnn}. 

\section{Proof of equivalence}
\label{app:proof_of_equivalence}
\begin{proof}
If we have a solution $\ket{\psi(t)}$ to Eq.~\ref{eq:2ndschro}
\begin{equation}
    \ket{\ddot{\psi}(t)} = - \bmsf{K}(t)\ket{\psi(t)} \,.
\end{equation}
The quantum state $\ket{\psi(t)}$ also satisfies
\begin{equation}
     \ket{\ddot{\psi}(t)}  = -i {\dot{\bm{\mathsf{H}}}}(t) \ket{\psi(t)} - \bm{\mathsf{H}}(t)^2 \ket{\psi(t)}\,.
\end{equation}
\begin{align}
    \ket{\dot{\psi}(t)}- \ket{\dot{\psi}(0)} &= \int_0^{t}-i {\dot{\bm{\mathsf{H}}}}(t) \ket{\psi(t)} - \bm{\mathsf{H}}(t)^2 \ket{\psi(t)} dt \nonumber \\
    &= \int_0^{t} -i {\dot{\bm{\mathsf{H}}}(t)} \sket{\psi(t)} -i{\bm{\mathsf{H}}(t)} \sket{\dot{\psi}(t)}dt \nonumber \\
    &= -i\int_0^{t}  {\dot{\bm{\mathsf{H}}}(t)} \sket{\psi(t)} + {\bm{\mathsf{H}}(t)} \sket{\dot{\psi}(t)}dt \,.
\end{align}
Let $T(t) = \bsf{H}(t)\sket{\psi(t)}$; it follows that
\begin{align}
    \ket{\dot{\psi}(t)}- \ket{\dot{\psi}(0)} 
    &= -i\int_0^{t} \dot{T}(t)dt = -i\bsf{H}(t)\sket{\psi(t)} + i \bsf{H}(0)\sket{\psi(0)} \,.
\end{align}
Therefore,
\begin{align}
    \ket{\dot{\psi}(t)} = -i\bsf{H}(t)\sket{\psi(t)} + i \bsf{H}(0)\sket{\psi(0)} +  \ket{\dot{\psi}(0)} \,.
\end{align}
We get back Eq.~\ref{eq:Schr} if $i \bsf{H}(0)\sket{\psi(0)} +  \ket{\dot{\psi}(0)}  = 0$, which is the Schr\"odinger equation at $t=0$. 
\end{proof}

\section{Calculating \texorpdfstring{$\bcal{K}(t)$}{bcalKt}}
\label{app:square}
We show here that there are two ways to derive $\bcal{K}(t)$. In both approaches, we are given a time-dependent Hamiltonian $\bsf{H}(t)$ and seek to derive $\bcal{K}(t)$ from it. Note that $\bsf{H}(t)$ is generally a time dependent complex matrix. 

The first approach is to calculate $\bsf{K}(t) = i\dot{\bsf{H}}(t) + \bsf{H}(t)^2$. Then proceed with 
\begin{equation}
\bsf{K}(t) \mapsto
\bcal{K}(t) = 
\begin{bmatrix}
  \Re{\bsf{K}(t)} & -\Im{\bsf{K}(t)} \\
  \Im{\bsf{K}(t)} & \Re{\bsf{K}(t)}
\end{bmatrix} \,.
\end{equation}
The second approach is to first turn $\bsf{H}(t)$ real before squaring, i.e.,
\begin{equation}
\bsf{H}(t) \mapsto
\bcal{H}(t) = 
\begin{bmatrix}
  \Re{\bsf{H}(t)} & -\Im{\bsf{H}(t)} \\
  \Im{\bsf{H}(t)} & \Re{\bsf{H}(t)}
\end{bmatrix} \,.
\end{equation}

We depict the two approaches in a flow chart (Fig.~\ref{fig:branchchart}).
\begin{figure}[htb!]
    \centering
    \includegraphics[width=0.77\columnwidth]{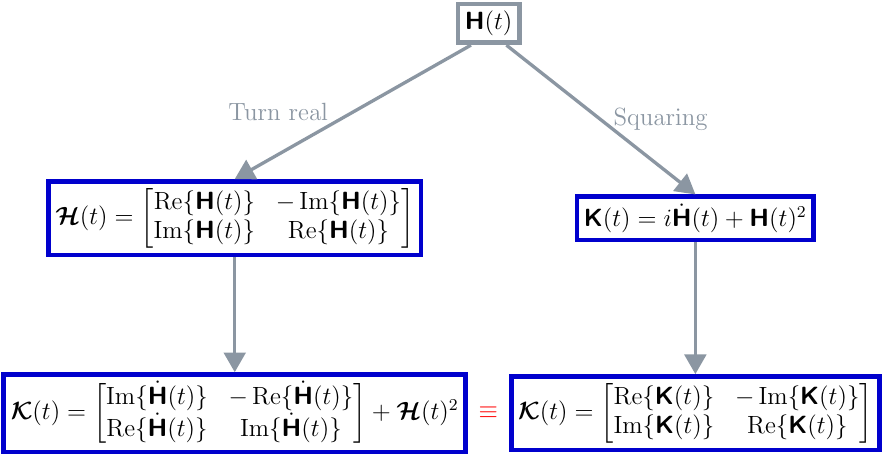}
    \caption{Two different approaches of calculating $\bcal{K}(t)$.}
    \label{fig:branchchart}
\end{figure}

The two approaches lead to the same $\bcal{K}(t)$.   

\section{Proof of equivalence: Real vs complex}
\label{app:proof_of_equivalence_real}
Recall that Eq.~\ref{eq:2ndschro} is 
\begin{equation}
\ket{\ddot{\psi}(t)} = - \bm{\mathsf{K}}(t)|\psi(t)\rangle \,.   
\end{equation}
We want to prove that Eq.~\ref{eq:realform},
\begin{align}
    \sket{\ddot{\varphi}(t)} + \bcal{K}(t)\sket{\varphi(t)} = 0
\end{align}
is equivalent to Eq.~\ref{eq:2ndschro}.
\begin{proof}
\begin{align}
    \sket{\ddot{\varphi}(t)} = \frac{d^2}{dt^2}\begin{bmatrix}
    \Re{\ket{\psi(t)}} \\
    \Im{\ket{\psi(t)}}
    \end{bmatrix}
    &= \begin{bmatrix}
    \Re{\ket{\ddot{\psi}(t)}} \\
    \Im{\ket{\ddot{\psi}(t)}}
    \end{bmatrix} \\
    &= -\begin{bmatrix}
    \Re{\bsf{K}(t)\sket{\psi(t)}} \\
    \Im{\bsf{K}(t)\sket{\psi(t)}}
    \end{bmatrix} \\
    &= -\begin{bmatrix}
    \Re{(\Re{\bsf{K}(t)} + i\Im{\bsf{K}(t)})(\Re{\sket{\psi(t)}}+i\Im{\sket{\psi(t)}}) } \\
    \Im{(\Re{\bsf{K}(t)} + i\Im{\bsf{K}(t)})(\Re{\sket{\psi(t)}}+i\Im{\sket{\psi(t)}}) }
    \end{bmatrix} \\
    &= -\begin{bmatrix}
    \Re{\bsf{K}(t)}\Re{\sket{\psi(t)}} - \Im{\bsf{K}(t)}\Im{\sket{\psi(t)}} \\
    \Im{\bsf{K}(t)}\Re{\sket{\psi(t)}} + \Re{\bsf{K}(t)}\Im{\sket{\psi(t)}}
    \end{bmatrix} \\
    &= -\left[\begin{array}{cc}
\Re{\bmsf{K}(t)} & -\Im{\bmsf{K}(t)} \\
\Im{\bmsf{K}(t)} & \Re{\bmsf{K}(t)}
\end{array}\right]\left[\begin{array}{c}
\Re{\ket{\psi(t)}} \\
\Im{\ket{\psi(t)}}
\end{array}\right] \\
&= - \bcal{K}(t) \sket{\varphi(t)}
\end{align}
\end{proof}

\section{Criterion for \texorpdfstring{$\bsf{K}(t)$}{Kt} to be real}
\label{app:real_symmetric_k}
Note that
\begin{align}
\label{eq:ktexpanded}
    \bsf{K}(t) &= i\dot{\bsf{H}}(t) + \bsf{H}(t)^2 \nonumber\\
    &= i\left(\Re{\dot{\bsf{H}}(t)} + i\Im{\dot{\bsf{H}}(t)}\right) + \Re{\bsf{H}(t)}^2 + i\Re{\bsf{H}(t)}\Im{\bsf{H}(t)} + i\Im{\bsf{H}(t)}\Re{\bsf{H}(t)} - \Im{\bsf{H}(t)}^2 \nonumber\\
    &= i\left(\Re{\dot{\bsf{H}}(t)} + \Re{\bsf{H}(t)}\Im{\bsf{H}(t)} + \Im{\bsf{H}(t)}\Re{\bsf{H}(t)}\right) + \Re{\bsf{H}(t)}^2 - \Im{\bsf{H}(t)}^2 - \Im{\dot{\bsf{H}}(t)} \,.
\end{align}
Therefore, $\bsf{K}(t)$ is real if and only if $\Re{\dot{\bsf{H}}(t)} + \Re{\bsf{H}(t)}\Im{\bsf{H}(t)} + \Im{\bsf{H}(t)}\Re{\bsf{H}(t)} = 0$. 
Because $\bsf{H}(t)$ is  Hermitian, the component $\Re{\bsf{H}(t)}$ is symmetric and $\Im{\bsf{H}(t)}$ is skew-symmetric. Note that 
\begin{equation*}
    \Re{\bsf{H}(t)}\Im{\bsf{H}(t)} + \Im{\bsf{H}(t)}\Re{\bsf{H}(t)}
\end{equation*}
is skew-symmetric because
\begin{align}
    (\Re{\bsf{H}(t)}\Im{\bsf{H}(t)} + \Im{\bsf{H}(t)}\Re{\bsf{H}(t)})^{T} &= \Im{\bsf{H}(t)}^{T}\Re{\bsf{H}(t)}^{T} + \Re{\bsf{H}(t)}^{T}\Im{\bsf{H}(t)}^{T} \nonumber\\
    &= (-\Im{\bsf{H}(t)})\Re{\bsf{H}(t)} + \Re{\bsf{H}(t)}(-\Im{\bsf{H}(t)}) \nonumber\\
    &= -(\Re{\bsf{H}(t)}\Im{\bsf{H}(t)} + \Im{\bsf{H}(t)}\Re{\bsf{H}(t)}) \,.
\end{align}
On the other hand, $\dot{\bsf{H}}(t)$ is Hermitian so $\Re{\dot{\bsf{H}}(t)}$ is symmetric. Therefore, \[\Re{\dot{\bsf{H}}(t)} + \Re{\bsf{H}(t)}\Im{\bsf{H}(t)} + \Im{\bsf{H}(t)}\Re{\bsf{H}(t)} = 0\] if and only if $\Re{\dot{\bsf{H}}(t)} = 0$ and $\Re{\bsf{H}(t)}\Im{\bsf{H}(t)} + \Im{\bsf{H}(t)}\Re{\bsf{H}(t)} = 0$.
The latter condition is also equivalent to the following:
\begin{align}
    &\phantom{{}\iff\kern 0.15557em}\Re{\bsf{H}(t)}\Im{\bsf{H}(t)} + \Im{\bsf{H}(t)}\Re{\bsf{H}(t)} = 0 \nonumber\\
    &\iff \Re{\bsf{H}(t)}\Im{\bsf{H}(t)} = -\Im{\bsf{H}(t)}\Re{\bsf{H}(t)} \nonumber\\
    &\iff \Re{\bsf{H}(t)}\Im{\bsf{H}(t)} = (\Re{\bsf{H}(t)}\Im{\bsf{H}(t)})^{T} \,,
\end{align}
i.e., $\Re{\bsf{H}(t)}\Im{\bsf{H}(t)}$ is symmetric.
Therefore, a real $\bsf{K}(t)$ is equivalent to:
\begin{enumerate}
    \item $\Re{\dot{\bsf{H}}(t)} = 0$.
    \item Symmetric $\Re{\bsf{H}(t)}\Im{\bsf{H}(t)}$.
\end{enumerate}

If further $\Im{\dot{\bsf{H}}(t)} = 0$ in Eq.~\ref{eq:ktexpanded}, the matrix $\bsf{K}(t)$ is real symmetric:
\begin{align}
    \bsf{K}(t) = (\Re{\bsf{H}(t)}^2 - \Im{\bsf{H}(t)}^2)^{T} 
    &= \Re{\bsf{H}(t)}^{T}\Re{\bsf{H}(t)}^{T} - \Im{\bsf{H}(t)}^{T}\Im{\bsf{H}(t)}^{T} \nonumber\\
    &= \Re{\bsf{H}(t)}\Re{\bsf{H}(t)} - (-\Im{\bsf{H}(t)})(-\Im{\bsf{H}(t)}) \nonumber\\
    &= \Re{\bsf{H}(t)}^2 - \Im{\bsf{H}(t)}^2 \,.
\end{align}

\subsection{Example}
\begin{itemize}
    \item Any real symmetric $\bsf{H}$ leads to a real symmetric $\bsf{K}$.
    \item For
\begin{equation}
\bsf{H}(t) = \left[\begin{array}{cc}
1 & 1-2ti \\
1+2ti & -1
\end{array}\right] \,,
\end{equation}
\begin{equation}
\bsf{K}(t) = \left[\begin{array}{cc}
2-4t^2 & 2 \\
-2 & 2-4t^2
\end{array}\right] \,,
\end{equation}
which is real but not symmetric.
    \item For 
\begin{equation}
\bsf{H} = \left[\begin{array}{ccc}
1 & -i & 2i\\
i & -1 & 0\\
-2i & 0 &  -1
\end{array}\right] \,,
\end{equation}
\begin{equation}
\bsf{K} = \left[\begin{array}{ccc}
6 & 0 & 0\\
0 & 2 & -2\\
0 & -2 &  5
\end{array}\right] \,,
\end{equation}
which is real and symmetric.
\end{itemize}

\section{Proof of \texorpdfstring{$\bcal{K}(t)$}{bcalKt} is real symmetric when \texorpdfstring{$\bsf{K}(t)$}{bsfKt} is Hermitian}
\label{app:simpleproof_2}
If $\bsf{K}(t) = \Re{\bsf{K}(t)} + i\Im{\bsf{K}(t)}$ is Hermitian, $\Re{\bsf{K}(t)}$ is symmetric and $\Im{\bsf{K}(t)}$ is skew-symmetric. The transpose of the block matrix $\bcal{K}(t)$ is
\begin{equation}
\bcal{K}(t)^{T} = \left[\begin{array}{cc}
\Re{\bm{\mathsf{K}}(t)} & -\Im{\bm{\mathsf{K}}(t)} \\
\Im{\bm{\mathsf{K}}(t)} & \Re{\bm{\mathsf{K}}(t)}
\end{array}\right]^{T} = \left[\begin{array}{cc}
\Re{\bm{\mathsf{K}}(t)}^{T} & \Im{\bm{\mathsf{K}}(t)}^{T} \\
-\Im{\bm{\mathsf{K}}(t)}^{T} & \Re{\bm{\mathsf{K}}(t)}^{T}
\end{array}\right] = \left[\begin{array}{cc}
\Re{\bm{\mathsf{K}}(t)} & -\Im{\bm{\mathsf{K}}(t)} \\
\Im{\bm{\mathsf{K}}(t)} & \Re{\bm{\mathsf{K}}(t)}
\end{array}\right] = \bcal{K}(t)\,.
\end{equation}
Therefore $\bcal{K}(t)$ is real symmetric.

\section{Proof of \texorpdfstring{$\bcal{K}(t)$}{bcalKt} is positive semidefinite when \texorpdfstring{$\bsf{K}(t)$}{bsfKt} is positive semidefinite}
\label{app:simpleproof_3}
Although in the main text we want to prove that a positive semidefinite $\bsf{K}$ implies a positive semidefinite $\bcal{K}$, we add time-dependence to the operators to make the proof more general. The proof holds at each instant $t$.

If $\bsf{K}(t)$ is positive semidefinite, 
$\braket{\psi(t)|\bsf{K}(t)|\psi(t)} \geq 0$ for all $\ket{\psi(t)} \neq 0$. This implies that
\begin{align}
&\phantom{\implies.} \braket{\psi(t)|\bsf{K}(t)|\psi(t)} \geq 0 \,\,\,\,\,\text{for all $\ket{\psi(t)}\neq 0$}\nonumber\\
&\implies (\Re{\bra{\psi(t)}}-i\Im{\bra{\psi(t)}})(\Re{\bsf{K}(t)} + i\Im{\bsf{K}(t)})(\Re{\ket{\psi(t)}}+i\Im{\ket{\psi(t)}}) \geq 0 \,\,\,\,\,\text{for all $\ket{\psi(t)}\neq 0$}\nonumber \\
&\implies \begin{bmatrix}
\Re{\ket{\psi(t)}} \hspace{1ex}
\Im{\ket{\psi(t)}}
\end{bmatrix}\hspace{-1ex}\begin{bmatrix}
\Re{\bsf{K}(t)} & -\Im{\bsf{K}(t)} \\
\Im{\bsf{K}(t)} & \Re{\bsf{K}(t)}
\end{bmatrix}
\hspace{-1ex}
\begin{bmatrix}
\Re{\ket{\psi(t)}} \\
\Im{\ket{\psi(t)}}
\end{bmatrix}
\geq 0 \,\,\,\text{for all $\ket{\psi(t)}\neq 0$}\nonumber \\
&\implies \braket{\varphi(t)|\bcal{K}(t)|\varphi(t)} \geq 0 \,\,\,\,\,\text{for all $\ket{\varphi(t)}\neq 0$ (Every $\ket{\varphi(t)}$ can be expressed as $\begin{bmatrix}
\Re{\ket{\psi(t)}} \\
\Im{\ket{\psi(t)}}
\end{bmatrix}$ for some $\ket{\psi(t)}$)}. 
\end{align}
Therefore, $\braket{\varphi(t)|\bcal{K}(t)|\varphi(t)} \geq 0$ for all $\ket{\varphi(t)} \neq 0$; i.e., $\bcal{K}(t)$ is positive semidefinite.

\section{Spring-mass equation as energy minimization}
\label{app:spring_mass_equation_derivation}
We derive the spring-mass equation from the Hamilton's principle. There are $N$ masses $m_{i}$, and $N(N+1)/2$ weightless springs each with spring constant $\kappa_{ij}$, where $1 \leq i,j \leq N$. Assuming that
the kinetic energy of the spring-mass system to be
\begin{equation}
   T = \frac{1}{2}\sum_i m_{i}\dot{x}_{i}(t)^2 \,,
\end{equation}
and that the potential energy is given by
\begin{equation}
   V =  \frac{1}{2}\sum_i\left(\sum_{j>=i}\kappa_{ij} x_{ij}(t)^2\right) \,.
\end{equation}
The Lagrangian of the spring-mass system is
\begin{equation}
\label{eq:springmasslagrangian_1}
    L = T - V = \frac{1}{2}\sum_i m_{ii}\dot{x}_{ii}(t)^2 - \frac{1}{2}\sum_i\left(\kappa_{ii} x_{ii}(t)^2 + \sum_{j>i}\kappa_{ij} x_{ij}(t)^2\right) \,.
\end{equation}

Here $x_{ij}(t) = x_i(t) - x_j(t)$, and $x_{ii}(t) = x_i(t)$.
Here $x_{ij}(t) = x_i(t) - (1-\delta_{ij})x_j(t)$, where $\delta_{ij} = 1$ if $i=j$, else $\delta_{ij} =0$. Eq.~\ref{eq:springmasslagrangian_1} can be rewritten as
\begin{equation}
\label{eq:springmasslagrangian_2}
    L = \frac{1}{2}\sum_i m_{ii}\dot{x}_{i}(t)^2 - \frac{1}{2}\sum_i\left(\kappa_{ii} x_{i}(t)^2 + \sum_{j>i}\kappa_{ij} (x_i(t)-x_j(t))^2\right)  \,.
\end{equation}
For each mass $i$, the Euler-Lagrange equation (E-L) equation is
\begin{equation}
    \frac{d}{dt}\left(\frac{\partial L}{\partial \dot{x}_i}\right) = \frac{\partial L}{\partial x_i} \,,
\end{equation}
--- which in this case is
\begin{align}
m_{ii}\ddot{x}_{i} &=  -\kappa_{ii}x_i -
\sum_{i > j}\kappa_{ji}(x_j - x_i)(-1) - 
\sum_{j > i}\kappa_{ij}(x_i - x_j)  \nonumber\\
&=  -\kappa_{ii}x_i -
\sum_{i \neq j}\kappa_{ij}(x_i - x_j) \quad\quad\quad\quad\quad\quad\quad\quad\quad\quad\text{(since $\kappa_{ij} = \kappa_{ji}$)}\,.
\end{align}

Thus we arrive at the following equation:
\begin{equation}
\label{eq:2ndorderODE_spring_appendix}
    \bsf{M}\ddot{\pmb{x}}(t) + \bsf{S}\pmb{x}(t) = 0\,,
\end{equation}
where $\bsf{M}_{ii} = m_{ii}$, $\bsf{M}_{ij} = 0$, $\bsf{S}_{ii} = \sum_{j}\kappa_{ij}$ and $\bsf{S}_{ij} = -\kappa_{ij}$.

\section{Numerical recipe for experiments}
\label{app:numericalrecipe}
From Eq.~\ref{eq:realform}, it is known that
\begin{equation}
\label{eq:dummy2nd_app}
    \sket{\ddot{\varphi}(t)} = -\bcal{K}(t)\sket{\varphi(t)}\,,
\end{equation}

For a digital classical computer however, it is necessary to turn higher order ODEs into $1$-st order ODEs. Therefore, we turn Eq.~\ref{eq:dummy2nd_app} into 

\begin{equation}
\label{eq:dummy1st_app}
    \sket{\dot{\Phi}\pars{t}} = -\bm{\mathrm{X}}(t)\sket{\Phi\pars{t}}\,,
\end{equation}
where 
\begin{equation}
\bm{\mathrm{X}}(t) = \begin{bmatrix}
\mathbf{0} & -\mathbf{I} \\
\bcal{K}(t)& \mathbf{0} \\
\end{bmatrix} 
\quad\mathrm{and}\quad
\sket{\dot{\Phi}(t)} = \begin{bmatrix}
\sket{\varphi(t)} \\
\sket{\dot{\varphi}(t)}
\end{bmatrix}
\,.
\end{equation}

With the initial conditions, one can solve Eq.~\ref{eq:dummy1st_app} and get the solution of $\sket{\varphi(t)}$ from the first half of the $\sket{\Phi\pars{t}}$ vector.

\section{Physical analogy}
\label{app:physics}
We suggest how to possibly construct some physical computing systems for the purpose of quantum simulation in Table~\ref{tab:physics}. 

\begin{table*}[htb!]
\caption{\label{tab:physics}Relations between different physical computing devices for analog classical simulation.}
\begin{tabular}{lllll}
 &  &   &  &  \\ 
 & Mechanical & Electrical  & Neuromorphic~\cite{Abernot} &  \\ \hline
$\bcal{K}(t)$ & Stiffness matrix & Capacitance matrix &  Neuron connectivity matrix &  \\ \hline
Analog unit & Mass displacement & Electrical charge& Neuron action potential \\ \hline
\end{tabular}
\end{table*}

Note that the analogy is not unique. For example, there exists multiple other mechanical–electrical analogies~\cite{firestone_new_1933}.

\end{document}